\documentclass[a4paper,UKenglish,cleveref, autoref, thm-restate]{lipics-v2021}

\pdfoutput=1 
\hideLIPIcs  

\author{Chris Barrett}{Department of Computer Science, University of Bath, United Kingdom}{}{}{}
\author{Willem Heijltjes}{Department of Computer Science, University of Bath, United Kingdom}{}{}{EPSRC Grant EP/R029121/1 \emph{Typed lambda-calculi with sharing and unsharing}}
\author{Guy McCusker}{Department of Computer Science, University of Bath, United Kingdom}{}{}{}

\authorrunning{C. Barrett, W. Heijltjes, G. McCusker}

\Copyright{Christopher Barrett, Willem Heijltjes, Guy McCusker}



\acknowledgements{Pierre Clairambault, Ugo Dal Lago, Anupam Das, Giulio Guerrieri, Jim Laird, Paul Levy, Simon Peyton Jones, John Power, Alex Simpson, Vincent van Oostrom.}

\nolinenumbers

\EventEditors{Bartek Klin and Elaine Pimentel}
\EventNoEds{2}
\EventLongTitle{31st EACSL Annual Conference on Computer Science Logic (CSL 2023)}
\EventShortTitle{CSL 2023}
\EventAcronym{CSL}
\EventYear{2023}
\EventDate{February 13--16, 2023}
\EventLocation{Warsaw, Poland}
\EventLogo{}
\SeriesVolume{252}
\ArticleNo{39}

\usepackage{amsmath,amssymb}
\usepackage{accents}
\usepackage{arydshln}
\usepackage{listings}
\usepackage{mathtools}
\usepackage{proof}
\usepackage{stmaryrd}
\usepackage{thmtools}
\usepackage{tikz}
	\usetikzlibrary{arrows.meta}
	\usetikzlibrary{calc}
\usepackage{xcolor}

\makeatletter

\newcommand\black{\color{black}}

\colorlet{0}{black}
\colorlet{1}{red!80!black}
\colorlet{2}{blue!80!black}
\colorlet{3}{orange!90!black}
\colorlet{4}{green!50!black}
\colorlet{5}{violet!80!black}
\colorlet{6}{teal!90!white}
\colorlet{7}{brown!90!black}
\colorlet{8}{pink!70!blue}
\colorlet{9}{lime!80!black}

\colorlet{term}{1}
\colorlet{type}{2}

\newcommand\colorterm{\color{term}}
\newcommand\colortype{\color{type}}

\newcommand\colored{%
  \let\color@term\colorterm%
  \let\color@type\colortype%
}
\newcommand\uncolored{%
  \let\color@term\relax%
  \let\color@type\relax%
  \black
}
\uncolored

\newcommand\vc[1]{\vcenter{\hbox{$#1$}}}

\newcommand\rvecup[1]{\accentset{\rightharpoonup}{#1}}
\newcommand\lvecup[1]{\accentset{\leftharpoonup}{#1}}

\newcommand\define[1]{\emph{#1}}
\newcommand\defeq{\stackrel{\scriptscriptstyle\Delta}=}

\newcommand\smallbin[1]{\mathchoice
      {\mathbin{\raise.2ex \hbox{$\scriptstyle      #1$}}}%
      {\mathbin{\raise.2ex \hbox{$\scriptstyle      #1$}}}%
      {\mathbin{\raise.12ex\hbox{$\scriptscriptstyle#1$}}}%
      {\mathbin{           \hbox{$\scriptscriptstyle#1$}}}}%

\newcommand\Con{\wedge}
\newcommand\Imp{\rightarrow}
\newcommand\Tim{\times}
\newcommand\Ten{\otimes}

\newcommand\con{\kern1pt{\smallbin\Con}\kern1pt}
\newcommand\imp{\kern1pt{\smallbin\Imp}\kern1pt}
\newcommand\tim{\kern1pt{\smallbin\Tim}\kern1pt}
\newcommand\ten{\kern1pt{\smallbin\Ten}\kern1pt}
\newcommand\arrimp{\kern1pt{\smallbin\rightsquigarrow}\kern1pt}

\newlength{\ilength}
\newcommand\+{\kern1pt{\smallbin+}\kern1pt}

\newcommand\fv[1]{\mathsf{fv}(#1)}
\newcommand\bv[1]{\mathsf{bv}(#1)}

\newcommand\uterm[1]{\underline{\term{#1}}}

\newcommand\seq@discard[1]{\seq@start}
\newcommand\seq@comma{\@ifnextchar,{{\black,\dots,}~\seq@discard}{{\black,}~\seq@start}}

\newcommand\vec@space{\kern0.4pt}

\newcommand\type[1]{{\colored\typ{#1}}}
\newcommand\typ[1]{%
  \vphantom)%
  \let\seq@start\type@start%
  \seq@start#1\@end%
}
\newcommand\type@start{\color@type\let\type@loop\type@next\type@loop}
\newcommand\type@next[1]{%
  \ifx#1\@end\let\type@loop\type@end\else%
  \ifx#1`\let\type@loop\type@escape\else%
  \ifx#1_\let\type@loop\type@sub\else%
  \ifx#1^\let\type@loop\type@sup\else%
  \ifx#1!\let\type@loop\type@rvec\else%
  \ifx#1?\let\type@loop\type@lvec\else%
  \ifx#1|\let\type@loop\type@bar\else%
  \ifx#1,\let\type@loop\seq@comma\else%
  \ifx#1>{\kern1pt\smallbin\Rightarrow\kern1pt}\else%
  \ifx#1.\kern1pt{\cdot}\kern1pt\else%
  \ifx#1A\!_A\else%
  \typevar{#1}%
  \fi\fi\fi\fi\fi\fi%
  \fi\fi\fi\fi\fi%
  \type@loop%
}
\newcommand\type@sub[1]{_{#1}\type@start}
\newcommand\type@sup[1]{^{{\type@sup@color #1}}\type@start}
\newcommand\type@rvec[1]{\vec@space\rvecup{\typevar{#1}\vec@space}\type@start}
\newcommand\type@lvec[1]{\vec@space\lvecup{\typevar{#1}\vec@space}\type@start}
\newcommand\type@bar{\@ifnextchar-{~{\black\vdash}~\seq@discard}{|\type@start}}
\newcommand\type@escape[1]{#1\type@start}
\newcommand\type@end{\uncolored}

\newcommand\typevar[1]{\typevar@start#1\relax}
\newcommand\typevar@start{\let\typevar@loop\typevar@next\typevar@loop}
\newcommand\typevar@next[1]{%
  \ifx#1\relax\let\typevar@loop\typevar@end\else%
  \ifx#1k\kappa\else%
  \ifx#1n\nu\else%
  \ifx#1p\pi\else%
  \ifx#1r\rho\else%
  \ifx#1s\sigma\else%
  \ifx#1t\tau\else%
  \ifx#1u\upsilon\else%
  \ifx#1w\omega\else%
  #1%
  \fi\fi\fi\fi\fi\fi\fi\fi\fi%
  \typevar@loop%
}
\newcommand\typevar@end{}

\newcounter{@parens}

\newcommand\term[1]{{\colored\trm{#1}}}
\newcommand\trm[1]{%
  \setcounter{@parens}{0}%
  \vphantom(%
  \let\seq@start\term@start%
  \seq@start#1\@end%
}
\newcommand\term@start{\color@term\let\term@loop=\term@next\term@loop}
\newcommand\term@next[1]{%
  \ifx#1\@end\let\term@loop\term@end\else%
  \ifx#1`\let\term@loop\term@escape\else%
  \ifx#1_\let\term@loop\term@sub\else%
  \ifx#1^\let\term@loop\term@sup\else%
  \ifx#1?\let\term@loop\term@lvec\else%
  \ifx#1!\let\term@loop\term@rvec\else%
  \ifx#1"\let\term@loop\term@color\else%
  \ifx#1:\let\term@loop\term@colon\else%
  \ifx#1|\let\term@loop\term@bar\else%
  \ifx#1((\stepcounter{@parens}\else%
  \ifx#1))\addtocounter{@parens}{-1}\else%
  \ifx#1\{\{\stepcounter{@parens}\else%
  \ifx#1\}\}\addtocounter{@parens}{-1}\else%
  \ifx#1,\ifnum\value{@parens}=0\let\term@loop\seq@comma\else,\fi\else%
  \ifx#1*\star\else%
  \ifx#1l\lambda\else%
  \ifx#1<\langle\else%
  \ifx#1>\rangle\else%
  \ifx#1..\,\else%
  \ifx#1;\,{;}\,\else%
  \ifx#1=\kern1pt{\smallbin=}\kern1pt\else%
  \ifx#1G{{\black\Gamma}}\else%
  \ifx#1D{{\black\Delta}}\else%
  \ifx#1\sim{\,\black\sim\,}\else%
  \ifx#1-\let\term@loop\term@dash\else%
  #1%
  \fi\fi\fi\fi%
  \fi\fi\fi\fi\fi%
  \fi\fi\fi\fi\fi%
  \fi\fi\fi\fi\fi%
  \fi\fi\fi\fi\fi%
  \fi%
  \term@loop%
}
\newcommand\term@end{\color{black}\uncolored}
\newcommand\term@discard[1]{\term@start}
\newcommand\term@escape[1]{#1\term@start}
\newcommand\term@color[1]{\color{#1}\term@start}
\newcommand\term@sub[1]{_{#1}\term@start}
\newcommand\term@sup[1]{^{#1}\term@start}
\newcommand\term@rvec[1]{\rvecup{#1}\term@start}
\newcommand\term@lvec[1]{\lvecup{#1}\term@start}
\newcommand\term@bar{\@ifnextchar-{~{\black\vdash}~\term@discard}{|\term@start}}
\newcommand\term@colon{\@ifnextchar={\term@assign}{\term@type}}
\newcommand\term@assign{\mathbin{{:}{=}}\term@discard}
\newcommand\term@type{\black\colon\type@start}
\newcommand\term@dash{\@ifnextchar>{\,{\black\rw}\,\term@discard}{-\term@start}}

\newcommand\SCatSim{\Lambda S/{\sim}}
\newcommand\SCatSimeqn{\Lambda S/{{\textsf{eqn}}}}

\newcommand\id{\mathsf{id}}

\newcommand\cat[1]{
  \cat@start#1\@end%
}
\newcommand\cat@start{\let\cat@loop\cat@next\cat@loop}
\newcommand\cat@next[1]{%
  \ifx#1\@end\let\cat@loop\cat@end\else%
  \ifx#1_\let\cat@loop\cat@sub\else%
  \ifx#1^\let\cat@loop\cat@sup\else%
  \ifx#1!\let\cat@loop\cat@rvec\else%
  \ifx#1?\let\cat@loop\cat@lvec\else%
  \ifx#1-\let\cat@loop\cat@arrows\else%
  \ifx#1*\tim\else%
  \ifx#1A\!_A\else%
  \typevar{#1}%
  \fi\fi\fi\fi\fi\fi%
  \fi\fi%
  \cat@loop%
}
\newcommand\cat@sub[1]{_{#1}\cat@start}
\newcommand\cat@sup[1]{^{#1}\cat@start}
\newcommand\cat@rvec[1]{\rvecup{\typevar{#1}}\cat@start}
\newcommand\cat@lvec[1]{\lvecup{\typevar{#1}}\cat@start}
\newcommand\cat@arrows{%
  \@ifnextchar-{\longrightarrow\cat@discard@two}{%
  \@ifnextchar>{\imp\cat@discard}{-\cat@start}}}
\newcommand\cat@end{}
\newcommand\cat@discard[1]{\cat@start}
\newcommand\cat@discard@two[2]{\cat@start}

\newcommand\rr[1]{%
  \ifx#1x\mathsf{var}\else%
  \ifx#1l\mathsf{abs}\else%
  \ifx#1a\mathsf{arg}\else%
  \ifx#1f\mathsf{fun}\else%
  \fi\fi\fi\fi%
}

\newcommand\get{\mathsf{get}}
\newcommand\set{\mathsf{set}}
\newcommand\print{\mathsf{write}}
\newcommand\rand{\mathsf{rand}}
\newcommand\rd{\mathsf{read}}

\newcommand\wrt{\mathsf{write}}

\newcommand\ifthen{\mathsf{if}}

\newcommand\Z{\mathbb Z}

\newlength\diamondsize
\newcommand\@diamond[2]{\tikz\draw[#2](-#1,0)--(0,#1)--(#1,0pt)--(0,-#1)--cycle;}
\newcommand\@ldiamond[2]{\tikz\draw[#2](-#1,0)--(0,#1)--(0,-#1)--cycle;}
\newcommand\@rdiamond[2]{\tikz\draw[#2](0,#1)--(#1,0pt)--(0,-#1)--cycle;}

\newcommand\filldiamond{\mathchoice%
  {\@diamond{2.0pt}{fill}}%
  {\@diamond{2.0pt}{fill}}%
  {\@diamond{1.5pt}{fill}}%
  {\@diamond{1.2pt}{fill}}}
\newcommand\leftfilltriangle{\mathchoice%
  {\@ldiamond{2.0pt}{fill}}%
  {\@ldiamond{2.0pt}{fill}}%
  {\@ldiamond{1.5pt}{fill}}%
  {\@ldiamond{1.2pt}{fill}}}
\newcommand\rightfilltriangle{\mathchoice%
  {\@rdiamond{2.0pt}{fill}}%
  {\@rdiamond{2.0pt}{fill}}%
  {\@rdiamond{1.5pt}{fill}}%
  {\@rdiamond{1.2pt}{fill}}}

\newcommand\opendiamond{\mathchoice%
  {\@diamond{2.0pt}{}}%
  {\@diamond{2.0pt}{}}%
  {\@diamond{1.5pt}{}}%
  {\@diamond{1.2pt}{}}}
\newcommand\leftopentriangle{\mathchoice%
  {\@ldiamond{2.0pt}{}}%
  {\@ldiamond{2.0pt}{}}%
  {\@ldiamond{1.5pt}{}}%
  {\@ldiamond{1.2pt}{}}}
\newcommand\rightopentriangle{\mathchoice%
  {\@rdiamond{2.0pt}{}}%
  {\@rdiamond{2.0pt}{}}%
  {\@rdiamond{1.5pt}{}}%
  {\@rdiamond{1.2pt}{}}}

\newcommand\llb\llbracket
\newcommand\rrb\rrbracket

\newcommand\llp{\llparenthesis\kern1pt}
\newcommand\rrp{\kern1pt\rrparenthesis}

\newcommand\floor[1]{\lfloor#1\rfloor}

\newcommand\e{\varepsilon}

\newcommand\loc[1]{\mathsf{loc}(\term{#1})}

\newcommand\rnd{\mathsf{rnd}}
\newcommand\nd {\mathsf{nd}}
\newcommand\inp{\mathsf{in}}
\newcommand\out{\mathsf{out}}

\newcommand\FMC{\textnormal{FMC}}

\tikzstyle{rwhead}=[>/.tip={Triangle[open,length=2.5pt,width=4.5pt]},|/.tip={Rectangle[length=.5pt,width=4.5pt]}]

\tikzstyle{rw} =[line width=.5pt,rwhead,->]
\tikzstyle{rws}=[line width=.5pt,rwhead,->.>]
\tikzstyle{rwn}=[line width=.5pt,rwhead,->.>|]
\tikzstyle{rwp}=[line width=.5pt,rwhead,->,double]
\tikzstyle{rwps}=[line width=.5pt,rwhead,->.>,double]

\newcommand\rw  {\mathrel{\tikz\draw[rw]  (0,0)--(10pt,0pt);}}

\newcommand\val{{v}}

\newcommand\ac{\kern1pt{\smallbin\ggg}\kern1pt}

\newcommand{\cppo}{\mathbf{Cppo}}
\newcommand{\sem}[1]{\llbracket #1 \rrbracket}
\newcommand{\loglt}{\mathbin{\lhd}}

\newcommand\sntype[1]{\sn{\type{#1}}}
\newcommand\snterm[1]{\sn{\term{#1}}}
\newcommand\N{\mathbb N}

\newcommand\TR[1]{\!\scriptstyle{\tr{#1}\vphantom{pb}}}%
\newcommand\tr[1]{%
  \ifx#1*\star\else%
  \ifx#1x\mathsf{var}\else%
  \ifx#1l\mathsf{abs}\else%
  \ifx#1a\mathsf{app}\else%
  \ifx#1c\mathsf{cut}\else%
  \ifx#1f\mathsf{base}\else%
  \ifx#1<\mathsf{lcut}\else%
  \ifx#1>\mathsf{rcut}\else%
  \fi\fi\fi\fi\fi\fi\fi\fi%
}

\newcommand\qr[1]{%
  \ifx#1*\star\else%
  \ifx#1x\mathsf{AX}\else%
  \ifx#1l\mathsf{ABS}\else%
  \ifx#1a\mathsf{APP}\else%
  \ifx#1c\mathsf{CUT}\else%
  \ifx#1f\mathsf{con}\else%
  \ifx#1<\mathsf{lcut}\else%
  \ifx#1>\mathsf{rcut}\else%
  \fi\fi\fi\fi\fi\fi\fi\fi%
}

\newcommand\sr[1]{\vphantom(%
  \ifx#1l{\imp}\!\mathsf R\else%
  \ifx#1a{\imp}\!\mathsf L\else%
  \ifx#1x\mathsf{ax}\else%
  \ifx#1c\mathsf{cut}\else%
  \ifx#1V{\forall}\mathsf R\else%
  \ifx#1A{\forall}\mathsf L\else%
  \ifx#1e\textsf{eta}\else%
  \ifx#1z\textsf{cls}%
  \fi\fi\fi\fi\fi\fi\fi\fi%
}

\newcommand\eval{\kern1pt{\Downarrow}\kern1pt}

\newcommand\result[2]{(#1,\term{#2}){\Downarrow}}

\newcommand\machine{\@ifstar\@textmachine\@mathmachine}

\newcommand\@textmachine[4]{(#1,\term{#2})\eval(#3,\term{#4})}

\newcommand\@mathmachine[4]{%
\begin{array}{@{(~}l@{~,~}r@{~)}}%
#1 & \term {#2}%
\\\hline\hline\rule[-5pt]{0pt}{15pt}%
#3 & \if#4*\term*\;\else\term{#4}\fi%
\end{array}%
}

\newcommand\step[4]{%
\begin{array}{@{(~}l@{~,~}r@{~)}}%
#1 & \term {#2}%
\\\hline%
#3 & \if#4*\term*\;\else\term{#4}\fi%
\end{array}%
}

\newcommand\diagdots[2][1.5,3]{
  \node[fill=black,circle,inner sep=0pt,minimum size=1.5pt] at ($(#2) - (#1)$) {};
  \node[fill=black,circle,inner sep=0pt,minimum size=1.5pt] at (#2) {};
  \node[fill=black,circle,inner sep=0pt,minimum size=1.5pt] at ($(#2) + (#1)$) {};
}
\tikzstyle{termbox}=[draw=term,fill=term!10,rounded corners,minimum size=20pt]
\tikzstyle{tallbox}=[draw=term,fill=term!10,rounded corners,minimum width=20pt,minimum height=40pt]
\tikzstyle{termpic}=[x=1pt,y=1pt,inner sep=0pt,outer sep=0pt,thick]

\newcommand\wires[5]{
  \node[anchor=east] at (#2,#3) {$\type{#1}$};
  \node[anchor=west] at (#4,#3) {$\type{#5}$};
  \draw ($(#2,#3) + ( 2, 6)$) -- ($(#4,#3) + (-2, 6)$);
  \draw ($(#2,#3) + ( 8,-6)$) -- ($(#4,#3) + (-8,-6)$);
  \diagdots[-1.5,3]{$(#2,#3) + (10,0)$}
  \diagdots[ 1.5,3]{$(#4,#3) - (10,0)$}
}

\newcommand\sn[1]{\llbracket #1\rrbracket}
\newcommand\collapse[1]{\lfloor #1\rfloor}
\newcommand\termint[1]{\{#1\}} 

\newcommand\concat[2]{#1 , #2}
\newcommand\verythinspace{\kern1pt}
\newcommand\from{\leftarrow}

\makeatother

\begin{document}

\title{The Functional Machine Calculus II: Semantics}

\maketitle

\begin{abstract}
The Functional Machine Calculus (FMC), recently introduced by the second author, is a generalization of the lambda-calculus which may faithfully encode the effects of higher-order mutable store, I/O and probabilistic/non-deterministic input. Significantly, it remains confluent and can be simply typed in the presence of these effects.

In this paper, we explore the denotational semantics of the FMC. We have three main contributions: first, we argue that its syntax -- in which both effects and lambda-calculus are realised using the same syntactic constructs -- is semantically natural, corresponding closely to the structure of a Scott-style domain theoretic semantics.
Second, we show that simple types confer strong normalization by extending Gandy's proof for the lambda-calculus, including a small simplification of the technique. Finally, we show that the typed FMC (without considering the specifics of encoded effects), modulo an appropriate equational theory, is a complete language for Cartesian closed categories.
\end{abstract}


\section{Introduction}

Almost without exception, modern programming languages support a combination of computational effects and higher-order mechanisms. Programmers and programming language theorists recognise that the effects and higher-order mechanisms are fundamentally different constructs, with radically different syntax and semantics: compare assignment and dereferencing operations with function definition and invocation, for example. In both operational and denotational accounts, the higher-order mechanism --- typically expressed in some variant of $\lambda$-calculus --- and the effects mechanisms are treated using distinct approaches, and indeed the combination of the two, not to mention the combination of multiple kinds of effects, requires careful handling.

In a previous paper~\cite{Heijltjes-2022} the second author introduced the Functional Machine Calculus (FMC), a compact programming language which eliminates these distinctions and supports higher-order effectful programming with a streamlined yet natural syntax and operational semantics. In this paper, we reprise the definition of the FMC and attempt to explain and explore its construction and behaviour from a denotational perspective. Beginning with a domain-theoretic analysis of computation with stacks, we discover that the $\lambda$-calculus and effectful aspects of programming can be viewed as being of exactly the same kind, in fact entirely interchangeable. The syntax and operational semantics of the FMC embodies the operations naturally supported by the denotational semantics, while remaining programmer-friendly: the prevous paper~\cite{Heijltjes-2022} demonstrates by example the wide range of effectful programs that can straightforwardly be expressed in this syntax.

The connection we exploit between domain-theoretic semantics of $\lambda$-calculus and an operational semantics based on stacks is not new. In~\cite{Streicher-Reus-1998}, Streicher and Reus show that Scott's well-known $D_\infty$ models motivate and explain the definition of Krivine's abstract machine for evaluating $\lambda$-terms. Scott's construction takes a domain $D$ and builds a model of the $\lambda$-calculus as an appropriate limit of a sequence of domains, shown on the left below. Streicher and Reus observed that this construction may alternatively be regarded as taking the limit of the sequence on the right, where each $D_n$ is recovered as $C_n \rightarrow D$, and $D_\infty$ is $C_\infty \rightarrow D$.
\[
\begin{array}[t]{rcl}
  D_0 &= & D \\
  D_{n+1} & = & D_n \rightarrow D_n
\end{array}\quad\quad\quad\quad
\begin{array}[t]{rcl}
  C_0 &= & 1 \\
  C_{n+1} & = & C_n \times (C_n \rightarrow D)
\end{array}
\]
With this view, the limit $C_\infty$ is a stream or (potentially infinite) stack of elements of $D_\infty$, i.e.\ denotations of $\lambda$-terms; and such terms consume a stack. The equations defining the interpretations of terms in this model show that the operations of application and abstraction correspond directly to pushing and popping from the stack. The return domain $D$ plays a very minor role in the semantics of $\lambda$-terms: in Streicher and Reus's view, it is the result type of continuations, and ordinary $\lambda$-terms never return. Our development in this paper adapts this in two ways. First, we choose the return domain $D$ to be the domain of stacks. Thus a term can be regarded as a \emph{stack transformer}, which immediately supports an operation of sequencing between terms, leading to a sublanguage of the FMC called the \emph{sequential $\lambda$-calculus}. Second, we enrich the domain equation with the familiar monad for global state, and observe that --- if states are also stacks --- the resulting domain equation is one of multiple-stack transformers, with \emph{no special status} afforded to the stack that implements $\lambda$-abstraction and application. Thus we arrive at the promised semantic explanation of the FMC, treating reader--writer style effects and higher-order mechanisms identically, and giving a denotational motivation for the stack-machine semantics.

We go on to study this calculus from a type-theoretic and category-theoretic perspective. The domain-theoretic model is untyped but, generalising the usual development for  $\lambda$-calculus, a simple type system can be imposed which describes the shapes of the stacks being operated upon. Adapting Gandy's proof technique for $\lambda$-calculus, we show that the well-typed terms of the FMC are strongly normalizing. Further, we study the categorical properties of the calculus, viewing terms as morphisms between types, and discover that, up to a notion of  contextual equivalence, well-typed programs form a Cartesian closed category. An equational theory is presented which refines this equivalence, and well-typed terms modulo this theory are shown to be a sound and complete language for CCCs. 

The properties identified above --- strong normalization, Cartesian closure, and the operational property of confluence established in~\cite{Heijltjes-2022} ---
are entirely expected of languages which are variants of typed $\lambda$-calculus, but perhaps surprising, or even shocking, in the setting of the effectful FMC. How can the combination of higher-order programming and effects, including state, input/output, nondeterminism and probability, retain properties of confluence and referential transparency?

We offer the following explanation. These properties are of the FMC as a general calculus, which remains close to the $\lambda$-calculus, independent of the encoding of effects. As explored in~\cite{Heijltjes-2022}, the FMC is confluent, and for encoded reader/writer effects this manifests as reduction following the algebraic laws of Plotkin and Power~\cite{Plotkin-Power-2002-FOSSACS}. The CCC semantics presented in this paper pertains to the FMC itself, and, we emphasize, is \emph{not} a semantics of effects: when particular properties of the encoded effects are taken into account, the semantics will no longer be a CCC. We explore this in more detail in Section~\ref{sec:categories}.


\section{The Functional Machine Calculus}

The Functional Machine Calculus (FMC) arises from an operational perspective on the $\lambda$-calculus, taking a simple call--by--name stack machine in the style of Krivine~\cite{Krivine-2007} as primary. The machine evaluates a term in the context of a stack, where application is a \emph{push} (of the argument) and abstraction is a \emph{pop}, binding the popped value to the abstracted variable. Since the stack machine is intended as an operational semantics, and not for implementation, for simplicity we use substitution rather than an environment. The FMC then introduces two natural generalizations.
\begin{description}

	\item[Locations]
The machine is generalized from one to multiple stacks or streams, indexed by a global set of \emph{locations} $A$. In the calculus, abstraction and application are parameterized in $A$ as pop- and push-actions on the associated stack. Stacks are homogeneous, but may be used to encode different reader/writer effects: an input stream (which may be non-deterministic or probabilistic), an output stream, or a global mutable variable.

	\item[Sequencing]
The calculus is extended with sequential composition of terms, which gives their consecutive execution on the machine, and an identity term as the unit to composition, which is the empty instruction sequence on the machine, analogous to imperative \emph{skip}. This generalizes the calculus from one of stack \emph{consumers} to stack \emph{transformers}, where a term may return multiple outputs to the stack, and gives control over execution, as demonstrated by the encoding of various reduction strategies including call--by--value and call--by--push--value (see~\cite{Heijltjes-2022}).

\end{description}
Locations are an innovation of the FMC~\cite{Heijltjes-2022}, while sequencing is familiar from Hasegawa's $\kappa$-calculus~\cite{Hasegawa-1995, Power-Thielecke-1999} and higher-order stack programming~(see e.g.\ \cite{Pestov-Ehrenberg-Groff-2010}). In the latter case, there are also certain similarities with Milner's action calculi \cite{DBLP:conf/mfcs/Milner93}. These two innovations are implemented technically as follows. To emphasize the operational intuition as \emph{push} and \emph{pop}, application $\term{M\,N}$ will be written as $\term{[N].M}$, and abstraction $\term{lx.M}$ as $\term{<x>.M}$. These are subsequently parameterized in \emph{locations} $\term a,\term b,\term c\in A$. \emph{Sequencing} introduces a \emph{nil} or \emph{skip} term $\term*$ and makes the variable construct a prefix $\term{x.M}$; sequential composition $\term{M;N}$ will be a defined operation.

\begin{definition}
\label{def:FMC}
The \define{Functional Machine Calculus} (\FMC) is given by the grammar
\begin{align*}
\term{M,N,P}
  \quad\Coloneqq\quad \term *
       ~\mid~ \term{x.M}
       ~\mid~ \term{[N]a.M}
       ~\mid~ \term{a<x>.M}
\end{align*}
where from left to right the \emph{term} constructors are \emph{nil}, a \emph{(sequential) variable}, an \emph{application} or \emph{push action} on the location $\term a$, and an \emph{abstraction} or \emph{pop action} on the location $\term a$ which binds $\term x$ in $\term M$. Terms are considered modulo $\alpha$-equivalence.
\end{definition}

\noindent
We may omit the trailing $\term{.*}$ from terms for readability. We will use a \emph{main} location $\term l$, omitted from the term notation, as the computation stack (as opposed to the stacks to interpret effects), on which (call--by--name) $\lambda$-terms embed. We will use constants as free variables; constant operators such as addition $\term +$ and conditionals $\term\ifthen$ will operate on $\term l$ as well.

\begin{example}
\label{ex:increment-countup}
Consider the following example terms.
\[
	\term{\rnd<x>.[x].c<y>.[y].+.<z>.[z]c} \qquad\qquad \term{[<x>.[x]\out.[x].[1].+].<f>.[0].f.f.f}
\]
The first term increments a memory cell $\term c$ with a random number. It pops $\term x$ from the random generator stream $\term\rnd$ and pushes it to the main stack; pops $\term y$ from the cell $\term c$ and pushes that to the main stack as well; then $\term +$ adds the top two elements of the main stack $x$ and $y$ pushing back the result $x+y$; and this is popped as $\term z$ and pushed back onto the cell $\term c$.

The second term counts up from zero to three, outputting $0,1,2$ and leaving $3$ on the main stack. The subterm $\term{<x>.[x]\out.[x].[1].+}$ pops $\term x$ from the main stack and sends it to the output location $\term\out$, and then $\term{[x].[1].+}$ leaves $x+1$ on the main stack. In the example, this term is pushed, popped as $\term f$, and called three times on zero.
\end{example}

\begin{definition}
\define{Composition} $\term{N;M}$ and \define{substitution} $\term{\{N/x\}M}$ are given by
\[
\begin{array}{rcll@{\qquad}rcll}
		 \term{*;M} &=&         \term M   &
&	\term{[P]a.N;M} &=& \term{[P]a.(N;M)} 
\\     \term{x.N;M} &=&    \term{x.(N;M)} &
&	\term{a<x>.N;M} &=& \term{a<x>.(N;M)} \qquad~ (\term x\notin\fv{\term M}) 
\\ \\
	\term{\{P/x\}*}      &=& \term *
&&	\term{\{P/x\}[N]a.M} &=& \term{[\{P/x\}N]a.\{P/x\}M}
\\	\term{\{P/x\}x.M}    &=& \term{P;\{P/x\}M}
&&	\term{\{P/x\}a<x>.M} &=& \term{a<x>.M}
\\	\term{\{P/x\}y.M}    &=& \term{y.\{P/x\}M} &
&	\term{\{P/x\}a<y>.M} &=& \term{a<y>.\{P/x\}M} \qquad (\term y\notin\fv{\term P})
\end{array}
\]
where, in the last two cases, $\term x\neq\term y$; both are capture-avoiding by the conditions $\term x\notin\fv{\term M}$ and $\term y\notin\fv{\term P}$, which can be satisfied by $\alpha$-conversion.
\end{definition}
\begin{lemma}\label{lem:assoc}
Composition is associative and has unit $\term{*}$. 
\end{lemma}

\begin{definition}
The \emph{functional abstract machine} is given by the following data. A \emph{stack} of terms $S$ is defined below left, and written with the top element to the right. A \emph{memory} $S_A$ is a family of stacks or streams in $A$, defined below right.
\[
	S ~\Coloneqq~\e~\mid~S{\cdot}\term M
\qquad
	S_A~\Coloneqq~\{\,S_a\,\mid\,a\in A\}
\]
The notation $S_A;S_a$ identifies the stack $S_a$ within $S_A$. A \emph{state} is a pair $(S_A,\term M)$ of a memory and a term. The \emph{transitions} or \emph{steps} of the machine are given below left as top--to--bottom rules. A \emph{run} of the machine is a sequence of steps, written as $(S_A,\term M)\eval(T_A,\term N)$ or with a double line as below right.
\[
\begin{array}{@{(~}l@{~,~}r@{~)}}
	S_A~;~S_a 	         & \term{[N]a.M}
\\\hline
	S_A~;~S_a{\cdot}\term N &      \term M
\end{array}
\qquad
\begin{array}{@{(~}l@{~,~}r@{~)}}
	S_A~;~S_a{\cdot}\term N &   \term{a<x>.M}
\\\hline
	S_A~;~S_a               & \term{\{N/x\}M}
\end{array}
\qquad\qquad\qquad
\machine{S_A}M{T_A}N
\]
\end{definition}

Constant operations such as addition $\term+$ and conditional $\term\ifthen$ pop the required number of items from the main stack and reinstate their result, as per the rule given below left.
The FMC then operates as a standard stack calculus: e.g.\ an arithmetic expression $1 + ((2 + 3) \times 4)$ is given as a term $\term{[4].[3].[2].+.\times.[1].+}$ which indeed returns $21$. Formally, a constant operator $\term{c^{n,m}}$ of arity $n,m$ is defined by a (partial) function $c^{n,m}$ from $n$ input terms to $m$ output terms, which generates the machine rule schema below right, where the output terms are put on the stack.
\[
	\step{S_A;S_\lambda\cdot\term 2\cdot\term 3}{+.M}{S_A;S_\lambda\cdot\term 5}M
\qquad
	\step{S_A;S_\lambda\cdot\term{N_n}\cdots\term{N_1}}{c^{n,m}.M}{S_A;S_\lambda\cdot\term{c^{n,m}(N_1,\dots,N_n)}}M \qquad \qquad
\]

\begin{example}
The first term of Example~\ref{ex:increment-countup} has the following run of the machine, where the $\term\rnd$ location is initialized with a stream with $\term{3}$ at the head, and $\term c$ with the value $\term{5}$.
\[
\scalebox{1}{$
\begin{array}{@{(~}l@{~;~}l@{~;~}l@{~~,~~}r@{~)}}
	     S_\rnd\cdot\term 3 & \e_c\cdot\term 5 & \e_\lambda                         & \term{\rnd<x>.[x].c<y>.[y].+.<z>.[z]c}
\\\hline S_\rnd             & \e_c\cdot\term 5 & \e_\lambda                         &         \term{[3].c<y>.[y].+.<z>.[z]c}
\\\hline S_\rnd             & \e_c\cdot\term 5 & \e_\lambda\cdot\term 3             &             \term{c<y>.[y].+.<z>.[z]c}
\\\hline S_\rnd             & \e_c             & \e_\lambda\cdot\term 3             &                  \term{[5].+.<z>.[z]c}
\\\hline S_\rnd             & \e_c             & \e_\lambda\cdot\term 3\cdot\term 5 &                      \term{+.<z>.[z]c}
\\\hline S_\rnd             & \e_c             & \e_\lambda\cdot\term 8             &                        \term{<z>.[z]c}
\\\hline S_\rnd             & \e_c             & \e_\lambda                         &                            \term{[8]c}
\\\hline S_\rnd             & \e_c\cdot\term 8 & \e_\lambda                         &                               \term{*}
\end{array}$}
\]
\end{example}

Beta-reduction in the $\lambda$-calculus, from the perspective of the machine, lets consecutive push and pop actions interact directly. With multiple stacks, these must be actions on the same location, while other locations may be accessed in-between. Informally, the reduction step will then be as follows, where the argument $\term{[N]a}$ and abstraction $\term{a<x>}$ may be separated by actions $\term{[P]b}$ and $\term{b<y>}$ with $\term a\neq\term b$: $\term{[N]a\dots a<x>.M}~\rw_\beta~\term{\dots\{N/x\}M}$. In addition, it must be avoided that any intervening $\term{b<y>}$ captures $\term y$ in $\term N$.

\begin{definition}\label{defn:laws}
A \define{head context} $\term H$ is a sequence of applications and abstractions terminating in a hole:
\[
	\term{H} ~\Coloneqq~ \term{\{\}}~\mid~\term{[M]a.H}~\mid~\term{a<x>.H}
\]
The term denoted $\term{H.M}$ is given by replacing the hole $\term{\{\}}$ in $\term{H}$ with $\term M$, where a binder $\term{a<x>}$ in $\term{H}$ captures in $\term M$. The \emph{binding variables} $\bv{\term H}$ of $\term H$ are those variables $\term x$ where $\term H$ is constructed over $\term{a<x>}$. The set of \define{locations} used in a term or context is denoted $\loc M$ respectively $\loc H$. Then \define{beta-reduction} and \define{eta-reduction} are defined respectively by the rewrite rule schemas below, where $\term a\notin\loc H$ for both reduction rules, $\bv{\term H}\cap\fv{\term N}=\varnothing$ for the $\beta$-rule, and $\term x\notin\bv{\term H}$ for the $\eta$-rule. Both reductions are closed under all contexts.
\[
	 \term{[N]a.H.a<x>.M}~\rw_\beta~\term{H.\{N/x\}M}  \qquad\qquad \term{a<x>.H.[x]a.M}~\rw_\eta~\term{H.M}\quad (x\notin\fv{\term M})
\]
\end{definition}
We now clarify the relationship between beta reduction and the machine. 
Evaluation of a term $\term{M}$ on the machine, given sufficient inputs in the form of a stack of terms $\term{N_1} \cdots \term{N_n}$, begins in the state $(\term{N_1} \cdots \term{N_n}, \term{M})$. The machine then implements a particular (weak) reduction strategy, with each \textit{pop} transition corresponding to a beta-reduction of the term $\term{[N_1]\ldots [N_n].M}$ corresponding to the machine state under evaluation. 

In Section \ref{sec:machine-equiv}, we further introduce a notion of observational equivalence based on the machine, dubbed \textit{machine equivalence}, where terms are considered equivalent if they send equivalent inputs to equivalent outputs. This is shown to validate the beta and eta equations in general. 

The two generalizations \emph{locations} and \emph{sequencing} are independent, and the two calculi that have one but not the other are of independent interest.
\begin{itemize}

	\item
The \emph{poly $\lambda$-calculus} has only \emph{locations}, and is given by the fragment below.
\[
	\term{M,N} ~\Coloneqq~ \term{x.*} ~\mid~ \term{[N]a.M} ~\mid~ \term{a<x>.M}
\]

	\item
The \emph{sequential $\lambda$-calculus} has only \emph{sequencing}, and is given by the fragment below. 
\[
	\term{M,N} ~\Coloneqq~ \term{*} ~\mid~ \term{x.M} ~\mid~ \term{[N].M} ~\mid~ \term{<x>.M}
\]
\end{itemize}


\begin{example}
To demonstrate how the FMC encodes both effects and evaluation strategies, we consider the following (standard) call--by--value $\lambda$-calculus with reader/writer effects. (We assume familiarity with the operational semantics of effects and call--by--value $\lambda$-calculus; for an introduction see Winskel~\cite{Winskel-1993}.)
\[
\begin{array}{r@{~}c@{~}l@{\qquad}l}
	M,N,P & \Coloneqq & x~\mid~M\,N~\mid~\lambda x.M 	& \emph{$\lambda$-calculus}
\\[0pt]		  & \mid & \rd~\mid~\wrt~ N;M 					& \emph{input/output}
\\[0pt]		  & \mid & c := N;M~\mid~!c  					& \emph{state update and lookup}
\\[0pt]		  & \mid & N\oplus M ~\mid~ N+M      			& \emph{probabilistic and non-deterministic sum}
\end{array}
\]
The full calculus is encoded into the FMC as follows. We let $A$ comprise the main location~$\term l$,  a location $\term\inp$ for input, $\term\out$ for output, $\term\rnd$ and $\term\nd$ for probabilistically and non-deterministically generated streams of boolean values $(\top,\bot)$, and one location for each global memory cell. A term $M$ encodes as $\term{M_\val}$ below, where $N+M$ encodes like $N\oplus M$ but with the stream $\term\nd$.
\[
\begin{aligned}
           \trm{x}_\val &= \term{[x]}
\\      \trm{(lx.M)}_\val &= \term{[<x>.M_\val]}
\\      \trm{(M\,N)}_\val &= \term{N_\val;M_\val;<x>.x}
\end{aligned}
\quad\!
\begin{aligned}
         \trm{\rd}_\val &= \term{\inp<x>.[x]}
\\       \trm{`!c}_\val &= \term{c<x>.[x]c.[c]}
\end{aligned}
\quad
\begin{aligned}
 	  \trm{(\wrt ~N;M)}_\val &= \term{N_\val.<x>.[x]\out.M_\val}
\\     \trm{(c:= N;M)}_\val &= \term{N_\val.<x>.c<\_>.[x]c.M_\val}
\\ \trm{(N\oplus M)}_\val &= \term{\rnd<x>.[N].[M].[x].\ifthen}
\end{aligned}
\]
We leave it to the reader to confirm that the interpretation generates the correct evaluation behaviour, and conclude the example with the encoding and reduction of the following term.
\[
\scalebox{0.9}{$
\begin{aligned}
	({(\lambda f. f(f\,0))\,(\lambda x.\wrt~x; !c)})_\val
  &~=~\term{[ <x>.}\uterm{[x].<v>}\term{.[v]\out.c<y>.[y]c.[y]].<f>.[0].[f].<z>.z.[f].<w>.w}
\\&~\rw~\term{[ <x>.[x]\out.c<y>.[y]c.[y]].<f>.[0].[f].<z>.z.}\uterm{[f].<w>}\term{.w}
\\&~\rw~\term{[ <x>.[x]\out.c<y>.[y]c.[y]].<f>.[0].}\uterm{[f].<z>}\term{.z.f}
\\&~\rw~\uterm{[ <x>.[x]\out.c<y>.[y]c.[y]].<f>}\term{.[0].f.f}
\\&~\rw~\term{[0]. <x>.[x]\out.c<y>.}\uterm{[y]c}\term{.[y]. <z>.[z]\out.}\uterm{c<w>}\term{.[w]c.[w]}
\\&~\rw~\term{[0]. <x>.[x]\out.c<y>.}\uterm{[y]. <z>}\term{.[z]\out.[y]c.[y]}
\\&~\rw~\uterm{[0]. <x>}\term{.[x]\out.c<y>.[y]\out.[y]c.[y]}
\\&~\rw~\term{[0]\out.c<y>.[y]\out.[y]c.[y]}
\end{aligned}$}
\]
\end{example}


\section{Domain-theoretic semantics}
\label{sec:domains}
Our aim in this section is to show that the syntax and stack-machine semantics of the FMC may be further justified by consideration of a simple domain-theoretic semantics.
We work in the category $\cppo$ of complete partial orders (with least-element) and continuous functions. We show that a domain equation for stack-transformers directly supports the operations of the sequential $\lambda$-calculus, directly extending a Scott-style semantics of $\lambda$-calculus. The step from sequential lambda-calculus to FMC is then nothing more than the incorporation of the state monad in the original domain equation. Our development has much in common with Streicher and Reus's work~\cite{Streicher-Reus-1998}; the key step in the move to the FMC is to allow computations to return a result --- a new stack --- which may be further operated upon by later computations, yielding the sequencing operation of the FMC.

We begin by constructing a domain $D$ to interpret terms. A stack can be seen as an element of $D^{\mathbb{N}}$. A term takes a stack and, after computation, returns a new stack as its result. We suppose that computations are modelled using an (unspecified) strong monad $T$ on $\cppo$; for now think of $T$ as the lifting monad. Then a term would be an element of a  domain satisfying $D \cong D^{\mathbb{N}} \rightarrow TD^{\mathbb{N}}$. This domain equation can be solved by standard techniques. Kleisli function composition gives rise to a sequencing operation $D \times D \rightarrow D$ which is associative, and forms a monoid with unit element given by the unit of the monad. We will write $(d_1.d_2)$ for the composition of two elements of $D$ using this operation.

Observe that $D$ is a \emph{reflexive object} in $\cppo$ and hence provides a model of the $\lambda$-calculus:
  \begin{align*}
    D  ~\cong~  D^{\mathbb{N}} \rightarrow TD^{\mathbb{N}}
       ~\cong~  (D^{\mathbb{N}} \times D)  \rightarrow TD^{\mathbb{N}}
       ~\cong~  D \rightarrow (D^{\mathbb{N}} \rightarrow TD^{\mathbb{N}})
       ~\cong~ D \rightarrow D.
  \end{align*}

We briefly spell out the semantics of $\lambda$-calculus induced by this model. Let $\rho$ range over \emph{environments}: functions from the set of variables to $D$. We use $s$ to range over $D^{\mathbb{N}}$; $(s\cdot d)$ denotes the stack resulting from pushing $d$ onto $s$. We shall elide the isomorphism $D \cong D^{\mathbb{N}} \rightarrow TD^{\mathbb{N}}$.
For any term $\term{M}$ and environment $\rho$ we define $\sem{\term{M}}\rho \in D$ (equivalently, a function $D^{\mathbb{N}} \rightarrow TD^{\mathbb{N}}$) as follows.
\begin{align*}
  \sem{\term{x.*}}\rho  =  \rho{x} \qquad
  \sem{\term{[N].M}}\rho~s  =  \sem{\term M}\rho~(s \cdot \sem{\term N}\rho)\qquad
  \sem{\term{<x>.M}}\rho~(s \cdot d)   =  \sem{\term M} \rho'~s
\end{align*}
where $\rho'(x) = d$ and  $\rho'(y) = \rho(y)$ for $y\not=x$.

These definitions show immediately that application is interpreted by pushing the argument onto the stack, and abstraction by popping a term from the stack. Thus this standard $\lambda$-calculus model directly justifies the machine semantics. It extends to the sequential $\lambda$-calculus by defining
\[
  \sem{\term{*}}\rho  =  \eta
\quad\quad
  \sem{\term{x.M}}\rho  =  (\rho(x)~.~\sem{\term M}\rho),
\]
where $\eta$ is the unit of the monad (and of the monoid).

Thanks to the compositionality of the semantics we can readily prove:
 \begin{lemma}
   $\sem{\term{M}}\rho = \sem{\term{M}}\rho'$ if $\forall x \in \fv{\term{M}}~\rho(x) = \rho'(x).$
 \end{lemma}
 As a consequence of this Lemma, we may speak of $\sem{\term{M}}$, independent of $\rho$, when $\term{M}$ is closed.

 We extend the semantics to stacks of closed terms. Suppose the monad $T$ is lifting. A stack $S$ denotes an element of $D^{\mathbb{N}}$:
\[
  \sem{\e}  =  \bot \quad\quad
  \sem{S \cdot \term M}  =  \langle \sem{S} , \sem{\term M}\rangle
\]
where we elide the isomorphism $D^{\mathbb{N}} \cong D^{\mathbb{N}} \times D$.

Thanks to the direct correspondence between the semantic equations and the machine transitions, we have:
\begin{lemma}[Soundness]
  Whenever $({S},\term M)\eval({T},\term N)$ (with all terms closed), $\sem{\term M}(\sem{S}) = \sem{\term N}(\sem{T})$.
\end{lemma}

\begin{theorem}[Adequacy]
  If $\sem{\term M}(\sem{S}) \neq \bot$ then there exists a $T$ such that $(S,\term M)\eval(T,\term*)$.
\end{theorem}
\begin{proof}
  The proof of this statement follows readily from the existence of three relations: a relation between elements of $D$ and closed terms; a relation between semantic streams in $D^{\mathbb{N}}$ and streams $S$; and a relation between computations in $TD$ and machine states $(S,\term{M})$. We write $\loglt$ for each of these relations, relying on the types to disambiguate. The relations are required to satisfy the following conditions:
  \begin{eqnarray*}
    d \loglt \term{M} & \text{iff} & \forall \sigma \in D^{\mathbb{N}}, \sigma \loglt S \Rightarrow d(\sigma) \loglt (S,\term{M})\\
    \sigma \loglt S & \text{iff} & \forall i. \sigma_i \loglt S_i\\
    k \loglt (S,\term{M}) & \text{iff} & k = \mathsf{lift}(\sigma) \Rightarrow (S,\term{M})\eval(T,\term *)~\text{and}~\sigma\loglt T.
  \end{eqnarray*}
  These conditions cannot be used as a definition of the relations $\loglt$, for example as a fixed point of an operator on such relations, because the first clause contains a negatively-occurring usage of $\loglt$. Nevertheless the existence of such relations can be established using standard techniques of denotational semantics. Pitts's work~\cite{Pitts-1996} gives an elegant general theory which enables the construction of such relations.

  Once $\loglt$ has been shown to exist, a straightforward induction on syntax establishes that for any term $\term{M}$, and any finite stream $S$, we have
  \[
    \sem{\term{M}} \loglt \term M \quad\quad \sem{S} \loglt S \quad\quad \sem{\term{M}}(\sem{S}) \loglt (S,\term{M})
  \]
  from which computational adequacy immediately follows.
\end{proof}
Note that our soundness and adequacy results are expressed in terms of the stack-machine evaluation mechanism. It is also the case that the denotational semantics validates the beta- and eta-laws of Definition \ref{defn:laws}, but our point in this section is to emphasise that the stack-machine semantics can be seen as an implementation of a natural denotational model.

Our denotational semantics so far gives an account of the sequential $\lambda$-calculus. To extend to the FMC, we replace the lifting monad with the \emph{state monad} $TX = St \rightarrow (St \times X)_\bot$, where $St$ is a domain of states. Our domain equation becomes
\[
  D  \cong  D^{\mathbb{N}} \rightarrow (St \rightarrow (St \times D^{\mathbb{N}})_\bot)
   \cong  St \times D^{\mathbb{N}} \rightarrow (St \times D^{\mathbb{N}})_\bot
\]

If we let $St = D^{\mathbb{N}}$, so that the values in the state are stacks, this is a domain equation for ``two-stack transformers''. As above, this is a reflexive object, now in \emph{two distinct ways} depending on which stack is used to interpret the arguments. This is exactly the FMC with two locations; extension to any finite set of locations is handled similarly, and the soundness and adequacy results may be proved in the same way. As we emphasized in the introduction, this semantics has the remarkable property that the stack used to interpret the operations of the $\lambda$-calculus has exactly the same status as that used to interpret state, and it is merely convention that distinguishes the two. This is precisely the point of view embodied by the novel syntax and operational semantics of the FMC.


\begin{figure}
\begin{align*}
	\infer[\TR *]{\term{G |- *:?tA>!tA}}{}
	\qquad
	\infer[\TR f]{\term{G, x:} \mathbb{\type{\alpha}} \term{|- x:}\mathbb{\type{\alpha}}}{}
	\qquad
	\infer[\TR l]
	  {\term{G |- a<x>.M : a(r)\,?sA>!tA}}
	  {\term{G , x:r |- M:?sA>!tA}}
\end{align*}
\begin{align*}
 	\infer[\TR a]
	  {\term{G |- [N]a.M: ?sA>!tA}}
	  {\term{G |- N:r} &&
	   \term{G |- M:a(r)\,?sA>!tA}
	  }
	\qquad
	\infer[\TR x]
	 {\term{G , x:?rA>!sA |- x.M:?rA\,?tA>!uA}}
	 {\term{G , x:?rA>!sA |- {\phantom{x.}}M:?sA\,?tA>!uA}}
\end{align*}
\caption{Typing rules for the Functional Machine Calculus}
\label{fig:FMC-types}
\end{figure}


\section{Simple types}
\label{sec:FMC types}

Simple types for the FMC~\cite{Heijltjes-2022} describe the input/output behaviour of the stack machine. The type system has three levels, mirroring the syntactic categories of the machine: types $\type t$ for terms $\term M$, type vectors $\type{!t}$ (or \emph{stack types}) for stacks $S$, and location-indexed families of type vectors $\type{!t_A}$ (or \emph{memory types}) for memories $S_A$. A function type is then an implication between an input memory type and an output memory type.

\begin{definition}
\emph{FMC-types} $\type{r,s,t,u}$ over a set of \emph{base types} $\Sigma$ are given by:
\[
		\type t	      ~\Coloneqq~ \type{\alpha} \in \Sigma ~\mid~\type{!sA>!tA}
\qquad	\type{!t\!_A} ~\Coloneqq~ \{\type{!t_a}\mid a\in A\}
\qquad	\type{!t}	  ~\Coloneqq~ \type{t_1\dots t_n}
\]
\end{definition}

Equivalently, one may view a function type as an implication between two vectors of location-indexed types, considered modulo the permutation of types on different locations.
\[
	\type{a_1(s_1)\dots a_n(s_n)~>~b_1(t_1) \dots b_m(t_m)}
\qquad\qquad
	\type{a(s)\,b(t)}\sim\type{b(t)\,a(s)}
\]
We introduce the following notation, which will enable us to write types also in the manner above. The \emph{empty} type vector is $\type\e$, and the empty memory type $\type{\e_A}$. A \emph{singleton} memory type $\type{a(!t)}$ is empty at every location except $\type a$, where it has $\type{!t}$: $\type{a(!t)_a}=\type{!t}$ and $\type{a(!t)_b}=\type{\e}$ for $\type a\neq \type b$. A singleton $\type{\lambda(!t)}$ on the main location $\lambda$ may be written as $\type{!t}$. \emph{Concatenation} of type vectors is denoted by juxtaposition and the \emph{reverse} of a type vector $\type{!t}=\type{t_1\dots t_n}$ is written $\type{?t}=\type{t_n\dots t_1}$. This extends point-wise to families, so $\type{!sA !tA}=\{\type{!s_a!t_a}\mid a\in A\}$ and $\type{?tA}=\{\type{?t_a}\mid a\in A\}$.

\begin{definition}
A \emph{judgement} $\term{G |- M:t}$ is a typed term in a \emph{context} $\Gamma=\term{x_1:t_1,,x_n:t_n}$, a finite function from variables to types. The typing rules for the FMC are given in Figure~\ref{fig:FMC-types}.
\end{definition}

\begin{example}
The terms from Example~\ref{ex:increment-countup} can be typed as follows, where $\type\Z$ is a base type of integers. Recall that the first term adds a random number to the cell $\term c$, and the second sends the numbers from zero to two to output, leaving the number three on the main stack.
\[
\begin{aligned}
	\term{\rnd<x>.[x].c<y>.[y].+.<z>.[z]c} &: \type{\rnd(\Z)\,c(\Z) > c(\Z)}
\\
	\term{[<x>.[x]\out.[x].[1].+].<f>.[0].f.f.f} &: \type{ > \out(\Z\,\Z\,\Z)\,\Z}
\end{aligned}
\]
\end{example}

Note, there are two typing rules for variables: one for variables of base type, and one for variables of higher type. The simply-typed FMC satisfies the subject reduction property \cite{Heijltjes-2022}, which is implicitly used in the following section. 


\section{Strong normalization}
\label{sec:SN}

We will show that reduction for the simply-typed FMC is strongly normalizing (SN). Our proof is a variant of Gandy's for the simply-typed $\lambda$-calculus~\cite{Gandy-1980}. Gandy's proof interprets terms in domains of strictly ordered, strict monotone functionals: the base domain is $\N^<=(\N,<_\N)$, and if $X=(|X|,<_X)$ and $Y=(|Y|,<_Y)$ are domains then so is $X\to Y$ where
\[
\begin{array}{rcl}
	|X\to Y| &{=}& \{ f \in Y^X \mid \forall x,x'\in X.~ x <_X x' \implies f(x)<_Y f(x') \}
\\
	f <_{X\to Y} g &{\iff}& \forall x\in X.~f(x) < g(x)~.
\end{array}
\]

The interpretation takes types to domains and terms of a given type to elements of that domain. The domains are well-founded and the interpretation of terms is such that it decreases on reduction, giving SN. One may further collapse a functional to a natural number to give an overestimate of the longest reduction sequence of a term. The literature has several variants on this proof, including one by De Vrijer that calculates longest reduction sequences exactly~\cite{DeVrijer-1987}; see also~\cite{Schwichtenberg-1982,VanDePol-Schwichtenberg-1995,VanDePol-1995}.

We introduce a (to the best of our knowledge) new variant, that avoids the domain of \emph{strict} functionals and instead interprets terms in the---more standard---domain of (non-strict) monotone functionals, as above but with $\leq$, generated from $\N^\leq=(\N,\leq_\N)$ with $\to$. This domain is not well-founded, but our interpretation of terms ensures that when functionals are collapsed to a natural number, this strictly decreases upon reduction, giving SN.

The technical difference is small and subtle. Gandy's proof originates in $\Lambda I$, where abstracted variables must occur, and hence the interpretation of an abstraction $\lambda x.M$ is naturally strict: the argument to $x$ always contributes to the overall interpretation. To generalize to the $\lambda$-calculus, where $x$ need not occur in $\lambda x.M$, a construction is introduced to nevertheless measure the argument to $x$, so that the functional for $\lambda x.M$ remains strictly monotone. The literature has several further such constructions~\cite{Gortz-Reuss-Sorensen-2003}.

This solves the challenge of accounting for reduction in terms that will be discarded, common in SN proofs. In our proof, instead we account for such terms when they are supplied as arguments: for a term $M\,N$ we increment the overall measure with that for $N$, measuring potential reduction in $N$ even if it will be discarded by $M$. An abstraction $\lambda x.M$ may then be interpreted as a standard monotone functional, avoiding strictness.

To build our domains, we use the $\to$ construction above, as well as the product of domains $X\times Y$ and an indexed product $\Pi_{a\in A}\,X_a$, defined in the expected way, as follows. Note, we will omit to work with base types in this section, so the base case is given by $\type{(>)}$. 
\[
\begin{aligned}
	|X\times Y| &= |X|\times|Y| \qquad
&
	(x,y)\leq_{X\times Y}(x',y') &\iff x\leq_X x' \wedge y\leq_Y y'
\\
	|\Pi_{a\in A}\,X_a| &= \Pi_{a\in A}\,|X_a|
&
	x\leq_{\Pi_{a\in A}X_a} x' & \iff \forall a\in A.~x_a\leq_{X_a}x'_a
\end{aligned}
\]

\begin{definition}
The \emph{interpretation} of an FMC-type $\type{t}$ is the domain $\sn{\type{t}}$ given by:
\[
\begin{aligned}
	\sntype{?sA>!tA} = \sntype{!sA}~\to\,\N^\leq\times\sntype{!tA}
\qquad
	\sntype{t_1\dots t_n} = \sntype{t_1}\times\dots\times\sntype{t_n}
\qquad
	\sntype{!tA} = \Pi_{a\in A}\,\sntype{!t_a}
\end{aligned}
\]
\end{definition}

It is worth observing that for the simple types of the $\lambda$-calculus, as embedded in the FMC, these domains are the natural ones. Briefly (see~\cite{Heijltjes-2022} for details), a simple type $\type{t_1 \imp \dots \imp t_n \imp o}$ embeds as the FMC-type $\type{t_1\dots t_n>}$ with the domain $\sntype{t_1}\times\dots\times\sntype{t_n}\to\N^\leq$, which is the expected one modulo Currying.

\begin{definition}
The least element of a domain $\sntype t$ is written $0_{\type t}$. The \emph{collapse} function $\collapse-_{\type t}:\sntype t\to\N$ takes a functional to a natural number by providing a least element as input and discarding all other output: $\collapse{f}_{\type{?sA>!tA}}=\pi_1(f(0_{\type{!sA}}))$.
\end{definition}

We will interpret terms such that if $\term{M:t}$ then $\snterm M\in\sntype t$, and if $\term M\rw\term N$ at type $\type t$ then both $\snterm M\geq_{\sntype t}\snterm N$ and $\collapse{\snterm M}>_\N\collapse{\snterm N}$, to give SN. We introduce the following notation.
To interpret terms in a context $\Gamma$, let a \emph{valuation} $v$ on $\Gamma$ be a function assigning to each variable $\term{x:t}$ in $\Gamma$ a value $v(\term{x})\in\sntype t$. The valuation $v\{\term{x}\leftarrow t\}$ assigns $t$ to $\term x$ and otherwise behaves as $v$.
We write elements of product domains as vectors $(t_1,\dots,t_n)$, and will elide the isomorphisms for associativity and unitality so that concatenation of $s$ and $t$ may be written $(s,t)$. Concatenation lifts to indexed products pointwise: $(s,t)_a = (s_a,t_a)$. For $t\in\sntype t$ we have a singleton $a(t)\in\sntype{a(t)}$ where $a(t)_a=t$ and $a(t)_b=()$ for $b\neq a$.

\begin{definition}
\label{def:sninterp}
For a term $\term{G |- M:t}$ and valuation $v$ on $\Gamma$, we inductively define the \emph{interpretation} $\snterm{G |- M:t}_v \in \sntype{t}$ as follows, omitting $\Gamma$ for compactness.
\[
\begin{array}{@{}r@{}l@{}l@{\hspace{-10pt}}r@{}l}
   \snterm{     *:       ?tA > !tA}_v & (t)      &{}= (0,   t)
\\ \snterm{   x.M:  ?rA\,?sA > !tA}_v & (s,r)    &{}= (n\+m,t)                & \text{where } (n,u) &{}= v(\term{x: ?rA > !uA})(r)
\\                                                                                        &&& (m,t) &{}= \snterm{M: ?uA\,?sA > !tA}_v (s,u)
\\ \snterm{a<x>.M: a(r)\,?sA > !tA}_v & (s,a(r)) &{}= (1\+m, t)               & \text{where } (m,t) &{}= \snterm{M:      ?sA > !tA}_{v\{\term{x} \leftarrow r\}}(s)
\\ \snterm{[N]a.M:       ?sA > !tA}_v & (s)      &{}= (1\+m\+\collapse{f}, t) & \text{where }     f &{}= \snterm{N:r}_v
\\                                                                                        &&& (m,t) &{}= \snterm{M: a(r)\,?sA > !tA}_v(s,a(f))
\end{array}
\]
We write $\snterm{G |- M:t}$ for $\snterm{G |- M:t}_v$ with $v$ the least valuation $v(\term{x:t})=0_{\type t}$, and may abbreviate $\snterm{G |- M:t}_v$ to $\snterm{M:t}_v$ or $\snterm{M}_v$.
\end{definition}

\begin{remark}
In this definition, the application case $\snterm{[N]a.M}_v(s)=(1\+m\+\collapse f)$ adds the value $\collapse f$ to account for reduction inside the argument $\term N$. Further, both it and the abstraction case $\snterm{a<x>.M}_v(s,a(r))=(1\+m,t)$ add $1$ to count redexes, so that a reduction step reduces the overall measure by (at least) $2$. It would suffice to count only abstractions or only applications, but the choice to count both is so that we count steps of the stack machine. We observe the following: for the alternative interpretation that omits to count $\collapse f$, and instead has $\snterm{[N]a.M}_v(s)=(1\+m)$, the collapsed interpretation $\collapse{\snterm M}$ measures the exact length of machine runs for $\term M$. This observation provides the proof with an operational intuition: terms are strongly normalizing because a) types guarantee termination of the machine~\cite[Theorem 3.12]{Heijltjes-2022}, and b) reduction shortens the length of machine runs.
\end{remark}

For the remainder of the proof, we will give an overview by stating the main lemmata. Each follows by a straightforward induction on typing derivations. 
First, for the interpretation $\sn{-}$ to be well-defined, the construction for each term must be shown to preserve monotonicity. We will do so in the following lemma. For valuations $v$ and $w$ over $\Gamma$, let $v\leq w$ denote that $v(\term x) \leq_{\sntype t} w(\term x)$ for all $\term{x:t}$ in $\Gamma$.
\begin{lemma}
\label{lem:increasing}
For all terms $\Gamma \vdash \term{M: t}$ and valuations $v\leq w$ over $\Gamma$, we have that:
\begin{enumerate}
	\item $\snterm{M}_v \in \sntype{t}$ 
	\item $\snterm{M}_v \leq_{\sntype{t}} \snterm{M}_w$.
\end{enumerate}
\end{lemma}

For the next steps, we first need that the interpretation commutes with sequential composition $\term{M;N}$ and substitution $\term{\{N/x\}M}$. Then, we show that reduction (non-strictly) decreases the interpretation, and strictly decreases the collapsed interpretation.

\begin{lemma}
\label{lem:sequencing}
For terms $\term{G |- M: ?sA?tA > !uA}$ and $\term{G |- N: ?rA > !sA}$ and valuation $v$ on $\Gamma$,
\begin{align*}
	\snterm{N;M}_v(t,r) = (i+j,u)
	\quad\text{where}\quad	\snterm{N}_v  (r)=(i,s) 
	\quad\text{and}  \quad  \snterm{M}_v(t,s)=(j,u)~.
\end{align*}
\end{lemma}

\begin{lemma}
\label{lem:substitution}
For terms $\term{G |- N:s}$ and $\term{G , x:s |- M:t}$ and valuation $v$ on $\Gamma$,
\[
	\snterm{\{N/x\}M}_v = \snterm M_{v\{\term x \from \snterm N_v\}}~.
\]
\end{lemma}

\begin{lemma}
\label{lem:reduction monotone}
If $\term{G |- M -> N : t}$ then $\snterm{M}_v \geq_{\sntype{t}} \snterm{N}_v$ for every valuation v on $\Gamma$.
\end{lemma}

\begin{lemma}
\label{lem:reduction collapsed}
If $\term{G |- M -> N: ?sA > !tA}$ then $\pi_1(\snterm M_v(s)) >_\N \pi_1(\snterm N_v(s))$ for every $s\in\sntype{!sA}$ and valuation $v$ on $\Gamma$.
\end{lemma}


The last lemma then immediately gives the strong normalization result.

\begin{theorem}[Strong Normalization]
\label{thm:sn}
Simply-typed FMC-terms are strongly normalizing with respect to beta-reduction.
\end{theorem}

\begin{proof}
By Lemma~\ref{lem:reduction collapsed} if $\term{G |- M -> N: t}$ then $\floor{\snterm M} >_\N \floor{\snterm N}$, so that $\floor{\snterm M}$ gives a bound for the length of any reduction path from $\term M$.
\end{proof}

Note that it is easy to extend this result to include eta-reduction: since eta-reduction does not increase the measure, and is clearly strongly normalizing by itself (the size of the term decreases), we can interleave each beta-reduction step with an arbitrary number of eta-reduction steps without affecting strong normalization.

\newcommand\ccceq{=_{\textsf{eqn}}}
\newcommand\ccc[1]{\llbracket#1\rrbracket}
\newcommand\cccterm[1]{\ccc{\term{#1}}}
\newcommand\ccctype[1]{\ccc{\type{#1}}}

	\section{Categorical semantics}
	\label{sec:categories}
	We give the categorical view on the FMC in three layers:
	\begin{itemize}
		\item terms with composition $\term{N;M}$ and unit $\term *$ form a category;
		\item terms modulo $\beta\eta$-equivalence form a \emph{premonoidal} category~\cite{Power-Robinson-1997};
		\item terms modulo an appropriate equational theory form a complete language for Cartesian closed categories;
	\end{itemize}
	we then show that \emph{machine equivalence}, where terms are equivalent if they display the same input/output behaviour on the machine, validates the final equational theory. 

The idea that a calculus with effects should semantically be a CCC may be surprising, so we will first motivate what the semantics does and does not capture. Firstly, and most importantly, the semantics we give here is one of the pure FMC, and emphatically \emph{not} a semantics \emph{of effects}: it ignores that, for instance, \emph{input} only has a \emph{pop} operation but no \emph{push}, and that \emph{state} locations would be restricted to a stack of depth one (at most). Imposing these constraints will cause the CCC semantics to break down, as we will demonstrate later in the section.

Secondly, the situation is analogous to the encoding of monadic effects in simply-typed $\lambda$-calculus, where for instance \emph{state} encodes as the monad $S\to (-\times S)$. In that case, too, the semantics remains a CCC, despite the possibility of encoding effects.

Thirdly, the two generalizations of the FMC, \emph{locations} and \emph{sequencing}, remain close enough to the $\lambda$-calculus that simple types allow to collapse them back onto a CCC semantics. \emph{Locations} give multiple copies of the $\lambda$-calculus, but because the types give the entire memory, the semantics may combine the different indexed stacks into one, projecting the multiple copies onto a single $\lambda$-calculus. \emph{Sequencing} gives control over evaluation and allows us to encode various reduction strategies, but the point of the denotational perspective is precisely to collapse any computational behaviour, and only consider the input/output behaviour of a term.

The purpose of our CCC semantics is to demonstrate that the simply-typed FMC is an \emph{operational} refinement of the lambda-calculus, but not a \emph{denotational} one. The FMC allows to express \emph{how} computation takes place: what reduction strategy is used, whether inputs are passed as function arguments or via mutable store, when the random generator is consulted, \emph{etc.} The denotational perspective then collapses these distinctions, demonstrating that we remain firmly in the domain of higher-order functional computation, despite the ability to encode effects.
\subsubsection*{The plain category}
For simplicity we will work in the sequential $\lambda$-calculus. The arguments generalize straightforwardly to the case of the FMC, and the details of this case are to be given in the first author's Ph.D. thesis. The objects are then type vectors $\type{!t}$ and morphisms in $\cat{!s --> !t}$ will be \textit{closed} terms $\term{M:?s>!t}$ modulo the given equivalence. 

A term $\term{M : r_1\dots r_m>s_n\dots s_1}$ may be represented by a string diagram as below, left. The wires represent the input and output stacks, with the first element at the top.
Strict composition of terms $\term{M:?r>!s}$ and $\term{N:?s>!t}$ into $\term{M;N:?r>!t}$, given below, right.
\[
	\qquad \qquad \qquad\scalebox{0.3}{\includegraphics{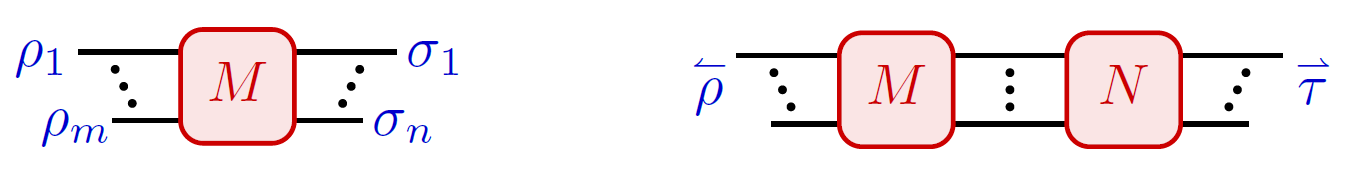}}
\]

	Analogous to the notation for type vectors, we  introduce the following notation.
	We use vector notation for variables, $\term{!x} = \term{x_1}\dots\term{x_n}$, and reverse a vector by a left-pointing arrow: if $\term{!x}$ is as before, then $\term{?x} = \term{x_n}\dots\term{x_1}$. Concatenation of vectors is by juxtaposition. We lift the notation to sequences of abstractions and applications, but only for variables: if $\term{!x}$ is as above, then $\term{<?x>.N} = \term{<x_n>\dots<x_1>.N}$ and $\term{[!x].N} = \term{[x_1]\dots[x_n].N}$. Vector notation is extended to contexts as $\term{x_1:t_1,,x_n:t_n}~=~\term{!x:!t}$  and  simultaneous substitutions as $\term{\{S/{}!x\}} = \term{\{M_1/x_1`,\dots`,M_n/x_n\}}$ , where  $S = \e\cdot\term{M_1}\cdots\term{M_n}$.

	The first, plain category is then given by Lemma \ref{lem:assoc}.

	\subsubsection*{The premonoidal category}

	A \emph{premonoidal} category~\cite{Power-Robinson-1997}, like a monoidal category, describes string diagrams, but with a \emph{sequential} element: a premonoidal product $\ten$ has no \emph{parallel} composition $f\ten g$, while $(f\ten\id);(\id\ten g)$ and $(\id\ten g);(f\ten\id)$ are distinct. 
	\[
\qquad \qquad \qquad 
\qquad \qquad 
	\vc{\begin{tikzpicture}[x=1pt,y=1pt,thick]
		\draw (0,20) -- (70,20);
		\draw (0, 0) -- (70, 0);
		\node[draw,fill=black!20,rounded corners,minimum size=18pt] at (20,20) {$f$};
		\node[draw,fill=black!20,rounded corners,minimum size=18pt] at (50, 0) {$g$};
	\end{tikzpicture}}
	\quad\neq\quad
	\vc{\begin{tikzpicture}[x=1pt,y=1pt,thick]
		\draw (0,20) -- (70,20) ;
		\draw (0, 0) -- (70, 0) ;
		\node[draw,fill=black!20,rounded corners,minimum size=18pt] at (50,20) {$f$};
		\node[draw,fill=black!20,rounded corners,minimum size=18pt] at (20, 0) {$g$};
	\end{tikzpicture}}
	\]
	Formally, a premonoidal product is a binary operation on objects $X\ten Y$ that is a functor in each argument, $-\ten X$ and $X\ten -$, but need not be a bifunctor $-\ten -$. 
In the FMC, the action on objects is concatenation, $\cat{!s\ten!t}~\defeq~\type{!s!t}$, with the first element at the top of the stack, and the unit given by $\e$. Both actions on morphisms are given below for $\term{M:?r>!s}$, where $\term{?x:?t}$.
	\[
	\begin{array}{rclc}
		\cat{M \ten !t :~ !t \ten !r --> !t \ten !s} &\defeq& \term{M:?r?t>!t!s} \qquad\qquad &
	\vc{\begin{tikzpicture}[termpic]
		\wires{?r}{-30}{20}{30}{!s}
		\wires{?t\,}{-20}{0}{20}{!t}
		\node[termbox] at (0,20) {$\term M$}; 
	\end{tikzpicture}}
	\\[2pt] \\ 
		\cat{!t \ten M :~ !r \ten !t --> !s \ten !t} &\defeq& \term{<?x>.M.[!x] : ?t?r > !s!t} \qquad\qquad &
	\vc{\begin{tikzpicture}[termpic]
		\node (a1) at (-28,36) {$\,\term{<x_1>}\,$};
		\node (b1) at ( 28,36) {$\,\term{[x_1]}\,$};
		\node (an) at (-22,24) {$\,\term{<x_n>}\,$};
		\node (bn) at ( 22,24) {$\,\term{[x_n]}\,$};
		\draw 
		  (-48,36)node[left] {$\type{t_1}\,$} -- (a1) (b1) -- ( 48,36)node[right]{$\,\type{t_1}$}
		  (-42,24)node[left] {$\type{t_n}\,$} -- (an) (bn) -- ( 42,24)node[right]{$\,\type{t_n}$};
		\draw[dotted,line cap=round,color=term] (a1)--(b1) (an)--(bn);
		\diagdots[-1.5,3]{-40,30}
		\diagdots[ 1.5,3]{ 40,30}
		\wires{?r}{-40}{10}{40}{!s}
		\node[termbox] at (0,10) {$\term M$}; 
	\end{tikzpicture}}
	\end{array}
	\]
	The first is \emph{expansion} (see Property 3.9 of the previous paper~\cite{Heijltjes-2022}). The second lifts the arguments for $\type{?t}$ from the stack as the variables $\term{?x}$, to restore them after evaluating $\term M$. We illustrate these above.
	A premonoidal product further has an \emph{associator} and a \emph{unitor}, and is called \emph{strict} if these are identities, which they are here. The category is then formed by terms modulo $\beta\eta$-equivalence, where $\eta$-equivalence is generated by:
	\[
		\term{M:r?s>!t}~=_\eta~\term{<x>.[x].M:r?s>!t} \quad\text{where }\term x\notin\fv{\term M}
	\]
	
	\begin{proposition}
	\label{prop:premonoidal}
	Terms modulo $\beta\eta$-equivalence form a strict premonoidal category.
	\end{proposition}
	We remark that terms modulo $\beta\eta$-equivalence do \textit{not} form a \textit{symmetric} pre-monoidal category, due to the failure of naturality of symmetry. Of course, one can add further equations to remedy this. In the sequel, we develop an extended equational theory which in fact makes the category of terms Cartesian closed. 
	
Premonoidal structure forces a notion of sequentiality, which has previously been employed to capture that of effects, as in the closed Freyd categories of Power, Thielecke, and Levy~\cite{Power-Thielecke-1999,Levy-Power-Thielecke-2003}, which are the premonoidal equivalent of Cartesian closed categories. However, this imposed sequentiality is only necessary if interactions through effects (such as state) are hidden from the type system. Because the FMC makes these explicit, they can instead be accounted for in the semantics, which then reverts to a Cartesian closed category.

\subsubsection*{The Cartesian closed category}

We now give an example illustrating why we would expect the FMC to form a Cartesian (closed) category, despite its ability to encode effects. For this example, but not for the rest of the section, we will then consider terms \textit{with} locations. Note, first, how the two following terms are illustrated in string diagrams below. 
	\[
\qquad \qquad 
	\term{<?x>}:~
	\vc{\begin{tikzpicture}[termpic]
		\node (a1) at (-28,36) {$\,\term{<x_1>}\,$};
		\node (an) at (-22,24) {$\,\term{<x_n>}\,$};
		\draw 
		  (-48,36)node[left] {$\type{t_1}\,$} -- (a1) 
		  (-42,24)node[left] {$\type{t_n}\,$} -- (an);
		\diagdots[-1.5,3]{-40,30}
	\end{tikzpicture}}
	\qquad\qquad
	\term{<?x>.[!x].[!x]}:~
	\vc{\begin{tikzpicture}[termpic]
		\node (a1) at (-28,36) {$\,\term{<x_1>}\,$};
		\node (an) at (-22,24) {$\,\term{<x_n>}\,$};
		\node (b1) at ( 34,36) {$\,\term{[x_1]}\,$};
		\node (bn) at ( 28,24) {$\,\term{[x_n]}\,$};
		\node (c1) at ( 22,12) {$\,\term{[x_1]}\,$};
		\node (cn) at ( 16, 0) {$\,\term{[x_n]}\,$};
		\draw 
		  (-48,36)node[left] {$\type{t_1}\,$} -- (a1) 
		  (-42,24)node[left] {$\type{t_n}\,$} -- (an) 
		  (b1) -- ( 54,36)node[right]{$\,\type{t_1}$}
		  (bn) -- ( 48,24)node[right]{$\,\type{t_n}$}
		  (c1) -- ( 42,12)node[right]{$\,\type{t_1}$}
		  (cn) -- ( 36, 0)node[right]{$\,\type{t_n}$};
		\draw[dotted,line cap=round,color=term] (a1)--(b1) (an)--(bn) (a1)--(c1) (an)--(cn);
		\diagdots[-1.5,3]{-40,30}
		\diagdots[ 1.5,3]{ 46,30}
		\diagdots[ 1.5,3]{ 34, 6}
	\end{tikzpicture}}
	\]
	\vspace{-\baselineskip}
	\begin{example}
	We introduce the following effectful operations as defined constructs (``sugar'') into the FMC: reading from a stream of random integers ($\type{\Z}$), $\term{\rand}$, a memory cell $\term c$ with $\term{\get}$ and $\term{\set}$ operations, and writing to output, $\term{\print}$.
	We give the definitions and types of these operations and illustrate these as as (nominal) string diagrams below, using colours and dashed lines to indicate non-main locations: \textcolor{red}{red} for $\term \rnd$, \textcolor{blue}{blue} for $\term c$ and \textcolor{orange}{yellow} for $\term \out$. We use the black dot here to depict a transition from one location to another\footnote{This corresponds to the \textit{renaming} a wire in the formalism of nominal string diagrams, \textit{i.e.}, string diagrams where wires additionally have an associated \textit{name}. } (as in $\term{\rnd}$, $\term{\set}$ and $\term{\print}$), as well as to depict the terminal $!$ (as in $\term{\set}$) and diagonal $\Delta$ (as in $\term{\get}$), as is standard.
	\[
	\begin{array}{cccc}
		\ \ \term{\rand}
		&\quad
		\term{\set}
		&\quad
		\term{\get}
		&
		 \term{\print}
	\\
		\ \ \term{\rnd<x>.[x]}
		&\quad
		\term{<x>.c<\_>.[x]c}
		&\quad
		\term{c<x>.[x]c.[x]}
		&
		 \term{<x>.[x]\out}
	\\
		\ \ \type{\rnd(\Z) > \Z}
		&\quad
		\type{\Z c(\Z) > c(\Z)}
		&\quad
		\type{c(\Z ) > c(\Z )\Z}
		&
		\type{\Z > \out(Z)}
	\\
		\begin{tikzpicture}[termpic]
		\filldraw[opacity=0, black] (0,13) circle (3pt);
		\node (a2) at (-15,30) {$\textcolor{red}{\rnd}$};
		\draw[red] [densely dashed](-5,30) -- (30,30) ;
		\draw (20,30) -- (45,30) ;
		\filldraw[black] (20,30) circle (3pt);
		\node[opacity=0.2,draw=term,fill=term!10,rounded corners,minimum size=20pt] at (20,30) {};
		\node (a1) at (50,30) {$\ \ \lambda$};
		\end{tikzpicture}
		&
		\begin{tikzpicture}[termpic]
		\node (a1) at (-5,20) {$\ \ \lambda$};
		\node (a1) at (-5,40) {$\ \ \textcolor{blue}{c}$} ;
		\draw[blue][densely dashed]  (5,40) -- (30,40) ;
		\draw (5,20) -- (20,20) ;
		\draw (20,20) -- (30, 25);
		\draw[blue][densely dashed] (30,25) -- (40,30) ;
		\draw[blue][densely dashed] (40,30) -- (55,30) ;
		\filldraw[black] (30,25) circle (3pt);
		\filldraw[black] (30,40) circle (3pt);
		\node[opacity=0.2,tallbox] at (30,30) {};
		\node (a1) at (55,30) {$\ \ \textcolor{blue}{c}$};
		\end{tikzpicture}
	&
		\begin{tikzpicture}[termpic]
		\node (a1) at (-5,30) {$\ \ \textcolor{blue}{c}$};
		\draw[blue][densely dashed]  (5,30) -- (30,30) ;
		\draw[blue][densely dashed] (40,20) -- (55,20) ;
		\draw[blue][densely dashed] (30,30) -- (40,20) ;
		\draw(40,40) -- (55,40) ;
		\draw(30,30) -- (40,40) ;
		\filldraw[black] (30,30) circle (3pt);
		\node[opacity=0.2,tallbox] at (30,30) {};
		\node (a1) at (55,40) {$\ \ \  \lambda$};
		\node (a1) at (55,20) {$\ \ \textcolor{blue}{c}$};
		\end{tikzpicture}
		& \quad
		\begin{tikzpicture}[termpic]
		\filldraw[opacity=0, black] (0,13) circle (3pt);
		\node (a2) at (-10,30) {$\lambda$};
		\draw(-5,30) -- (20,30) ;
		\draw [orange][densely dashed](20,30) -- (45,30) ;
		\filldraw[black] (20,30) circle (3pt);
		\node[opacity=0.2,draw=term,fill=term!10,rounded corners,minimum size=20pt] at (20,30) {};
		\node (a1) at (50,30){$\textcolor{orange}{\ \  \out}$};
		\end{tikzpicture}
	\end{array}
	\]
Note that, modulo renaming of wires, these effectful operations are encoded by the diagonal and terminal operations of a Cartesian category.
Consider the following term (reprised from Examples~$4.4$ and $4.8$ of the previous paper~\cite{Heijltjes-2022}) and its string diagrammatic representation.  Note that we give the diagram modulo symmetry.\footnote{Note that, because operations acting on different locations permute, we have, for example, no need to lift (and restore) the arguments from the stacks at locations $\term \lambda$ and $\term c$ prior to (and subsequent to) applying the second instances of $\term \rand$, $\term \set$ or $\term \get$ --- or, equivalently, doing so has the same result as not doing so. Technically, the term corresponding to the diagram would lift  the result of the first instance of $\term \get$ off the main $\term \lambda$ stack before applying the second instance of $\term \get$, and then restore it afterwards; however, since $\term + $ is symmetric in its inputs, we omit this for simplicity.}

\[
\scalebox{1.2}{$
\begin{array}{c}
	\term{~\rand~;\ ~\set~;~~\get~;~~\rand~;~~\set~;~~\get~;~+~;~~\print~ } \\ \\
	\vc{\begin{tikzpicture}[termpic]
		\node (a1) at (-10,40) {$\ \ \textcolor{blue}{c}$};
		\node (a2) at (-10,20) {$\textcolor{red}{\rnd}$};
		\node (a3) at (-10,0) {$\textcolor{red}{\rnd}$};
		\draw[blue][densely dashed]  (0,40) -- (50,40) ;
		\draw[red] [densely dashed](0,20) -- (30,20) ;
		\draw (20,20) -- (40,20) ;
		\draw (40,20) -- (50,25) ;
		\draw[red][densely dashed] (0, 0) -- (110, 0) ;
		\draw[blue][densely dashed] (50,25) -- (60,30) ;
		\draw[blue][densely dashed] (60,30) -- (80,30) ;
		\draw (90,40) -- (210,40) ;
		\draw (80,30) -- (90,40) ;
		\draw[blue][densely dashed] (80,30) -- (90,20) ;
		\draw[blue][densely dashed] (90,20) -- (140,20) ;
		\draw (110,0) -- (130,0) ;
		\draw (130,0)-- (140,5)  ;
		\draw[blue][densely dashed] (140, 5) -- (150,10) ;
		\draw[blue][densely dashed] (150,10) -- (170,10) ;
		\draw (180,20) -- (210,20) ;
		\draw (170,10) -- (180,20) ;
		\draw[blue][densely dashed] (170,10) -- (180,0) ;
		\draw[blue][densely dashed] (180,0) -- (250,0) ;
		\draw (210,30) -- (230,30) ;
		\draw[orange][densely dashed] (230,30) -- (250,30) ;
		\node[opacity=0.2,draw=term,fill=term!10,rounded corners,minimum size=20pt] at (20,20) {};
		\filldraw[black] (20,20) circle (3pt);
		\node[opacity=0.2, tallbox] at (50,30) {};
		\filldraw[black] (50,40) circle (3pt);
		\filldraw[black] (50,25) circle (3pt);
		\node[opacity=0.2,tallbox] at (80,30) {};
		\filldraw[black] (80,30) circle (3pt);
		\node[opacity=0.2,termbox] at (110, 0) {};
		\filldraw[black] (110,0) circle (3pt);
		\node[opacity=0.2,tallbox] at (140,10) {};
		\filldraw[black] (140,20) circle (3pt);
		\filldraw[black] (140,5) circle (3pt);
		\node[opacity=0.2,tallbox] at (170,10) {};
		\filldraw[black] (170,10) circle (3pt);
		\node[tallbox] at (200,30) {$\textbf{\term+}$};
		\node[opacity=0.2,termbox] at (230,30) {};
		\filldraw[black] (230,30) circle (3pt);
		\node (a3) at (253,0) {$\textcolor{blue}{c}$};
		\node (a3) at (255,30) {$\textcolor{orange}{\ \  \out}$};
	\end{tikzpicture}}
	\end{array}$}
\]

\bigskip

Due to the strong type system, we see all the dependencies between operations. 
For example, the second call to $\term \rand$ may safely be made before the first calls to $\term \set$ and $\term \get$. This would be illustrated by `sliding' the $\term{\rand}$ operation along the wire --- something which is forbidden in general in a pre-monoidal category, but is permissible in a monoidal setting. We can also see that the second $\term \set$ is dependent on the first $\term \get$. Indeed, their composition forms a beta-redex, corresponding to the expected interaction of $!$ with $\Delta$ in a Cartesian category.

Note, one can see in the above example that applying the naturality of the diagonal would result in a duplication on the location $\term \rnd$. This relies on $\term \rnd$ being a location with no special status, and in particular, having an associated \textit{push} action, similar to the main location $\term \lambda$. If we were to enforce that $\term \rnd$ was a read-only stream, then this duplication would no longer be possible and the semantics can no longer be Cartesian. Similar issues arise for memory cells, which ought to have depth (at most) one. We leave consideration of the particular properties of encoded effects for future work. 
	\end{example} 

	The following data will make $\ten$ a Cartesian product $\tim$, when considered modulo the following equational theory.
	The \emph{exponent} $\cat{!s->!t}~\defeq~\type{?s>!t}$ will then give Cartesian closure.
	\[
	\begin{array}{@{}r@{~}llr@{\,}l@{}}
	     	        ! :& \cat{!t --> 1}        &=& \term{<?x>}           &: \type{?t\,>}
	\\[2pt]	   \delta :& \cat{!t --> !t * !t}  &=& \term{<?x>.[!x].[!x]} &: \type{?t>!t!t}
	\\[2pt]	    \pi_1 :& \cat{!u * !t --> !t}  &=& \term{<?x>.<?y>.[!x]} &: \type{?t?u>!t}
	\\[2pt]	    \pi_2 :& \cat{!u* !t --> !u}  &=& \term{<?x>.<?y>.[!y]} &: \type{?t?u>!u}
	\\[2pt]    \epsilon   &: \cat{!s * (!s->!t) --> !t}   &=& \term{<z>.z}       &: \type{(?s>!t)?s>!t}
	\\[2pt] \eta       &: \cat{!t --> (!s -> !s * !t)}   &=& \term{<?x>.[[!x]]} &: \type{?t>(?s>!s!t)}
	\\[2pt] \cat{M->N} &: \cat{(!s->!t) --> (!r->!u)}  &=& \term{<z>.[M.z.N]} &: \type{(?s>!t)>(?r>!u)}
	\end{array}
	\]
where $\term{!x: !t}$, $\term{!y: !u}$, $\term{z: ?s > !t}$
\begin{definition}
We define the \textit{equational theory} $\ccceq$ of the FMC to be the least equivalence generated by the following laws, closed under all contexts. 
\begin{align*}
&\textup{Beta:} & \term{[N].<x>.M} &=_\beta \term{M\{N/x\}} && \type{?s > !t}\\ 
&\textup{Interchange:} &\term{<?x>.N.[!x].M} &=_\iota \term{M.<?y>.N.[!y]} && \type{?s?r > !u!t}\\
&\textup{Diagonal:} &\term{M<?y>.[!y].[!y]} &=_\Delta \term{<?x>.[!x].M.[!x].M} && \type{?s > !t!t} \\ 
&\textup{Terminal:} &\term{M.<?y>} &=_! \term{<?x>} && \type{?s > } \\
&\textup{Eta (First-order):} &\term{*} &=_\eta \term{<a>.[a]} && \type{\alpha > \alpha} \\
&\textup{Eta (Higher-order):} &\term{P} &=_\epsilon \term{<?x>.[[!x].P.<z>.z]} && \type{?r > (?s > !t)} 
\end{align*}where $\term{a: \alpha}$, $\term{!x : !s}, \term{!y: !t}, \term{M: ?s > !t}$, $\term{N: ?r > !u}$, $\term{z: ?s > !t}$ and $\term{P: ?r > (?s > !t)}$, and we do not allow abstractions to capture in $\term{M}, \term{N}$, or $\term{P}$. 
\end{definition}
Note that this theory includes beta- and eta-reduction. To see it includes eta-reduction at higher-type, consider the higher-order eta equation with $\term{P} = \term{*}$.\footnote{Following the definition of substitution, given a context $\term{\{}\!-\!\term{\}.M}$ with hole $\term{\{}\!-\!\term{\}}$, the substitution of a term $\term{N}$ into the hole is given by $\term{N;M}$, in particular with $\term{N}$ not binding in $\term{M}$. This means the eta equations together give $\term{<x>.[x].M} = \term{M}$, where $\term x \not\in \fv{\term{M}}$. }
\begin{theorem}\label{thm:sound-ccc-eqns}
Terms modulo $=_{\textsf{eqn}}$ form a strict Cartesian closed category.
\end{theorem}
The proof of the theorem above provides a canonical functor from  the free Cartesian closed category generated over a set of base types $\Sigma$, denoted $\textup{\textsf{CCC}}(\Sigma)$, to the category of FMC terms generated over th same signature, denoted $\SCatSimeqn$. We construct a left-inverse CCC-functor interpreting FMC-terms into $\lambda$-terms (with products and patterns), thus proving completeness.  In general, all constructions also work with also with constants drawn from a monoidal signature, as well as simply with a signature given by a set of base types.

The interpretation $\ccc{-}: \SCatSimeqn  \to \textup{\textsf{CCC}}(\Sigma)$ preserves types.
The top-level arrow $\type{>}$ of FMC-types becomes sequent entailment $\vdash$: the type of the input stack becomes the type of the $\lambda$-context and the type of the output stack becomes the type of the $\lambda$-term.

A \textit{valuation $v$} is a function assigning to each FMC-variable $\term{x:t}$ a $\lambda$-term $v(\term{x}) \in \ccctype{t}$. 
Given a valuation $v$, let $v\{\term{x}\leftarrow t\}$ denote the valuation which assigns $t$ to $\term x$ and otherwise behaves as $v$. 
We write contexts and products as vectors and elide the isomorphisms for associativity and unitality so that concatenation of $s$ and $t$ may be written as $s \cdot t$. 
\begin{definition}\label{defn:interp-to-lambda}
For each valuation $v$, define on the type derivation of $\term{G |- M: ?s > !t}$ an open $\lambda$-term 
$\cccterm{G |- M: ?s > !t}_v$, given by its action on contexts:
\begin{align*}
\ccc{\Gamma \vdash \term{*: ?s > !s}}_v&( s)&=  &\  s \\
\ccc{\Gamma \vdash \term{x: \alpha}}_v&&=  &\  v(\term x) \\
\ccc{\Gamma \vdash \term{<x>.M: r?s > !t}}_v&( s \cdot r)&=  &\ \ccc{\Gamma, \term{x: r} \vdash \term{M: ?s > !t}}_{v\{\term{x} \leftarrow r\}} ( s) \\
\ccc{\Gamma \vdash \term{[N].M: ?s > !t}}_v&( s)&=  &\ \ccc{\Gamma \vdash \term{M: r?s >!t}}_v ( s \cdot \ccc{\Gamma \vdash \term{N: r}}_v) \\
\ccc{\Gamma, \term{x: ?r > !u} \vdash \term{x.M: ?r?s > !t}}_v&( s \cdot  r)&= &\ \ccc{\Gamma,\term{x: ?r > !u} \vdash \term{M: ?u?s > !t}}_v ( s \cdot v(\term{x})( r)) 
\end{align*}  
\end{definition}

	\begin{theorem}
	\label{thm:ccc-eqns}
	Terms modulo $=_{\textsf{eqn}}$ form a complete language for Cartesian closed categories.
	\end{theorem}

\subsubsection*{Machine Equivalence}\label{sec:machine-equiv}
	A  natural contextual equivalence on terms is given by \emph{machine equivalence}, defined inductively on types below. It resembles the logical relation for program equivalence of Pitts and Stark~\cite{Pitts-Stark-1998}. We write $\result SM$ for $T$ if $(S,\term M)\eval(T,\term *)$ and take here the only constants to be of base type.
	
	\begin{definition}
	Closed terms $\term{M:?s>!t}$ and $\term{M':?s>!t}$ are \emph{machine equivalent at type $\type{?s>!t}$} if for equivalent inputs the machine gives equivalent outputs,
	\[
		\term{M\sim M': ?s>!t}
	~\defeq~
		\forall\, S{\sim}S':\type{!s}.~ \result SM \sim \result{S'}{M'} : \type{!t}
	\]
	where two terms of base type are equivalent if they are equal, and two stacks are equivalent if their terms are pairwise equivalent. Equivalence extends to open terms $\term{G |- M:t}$ and $\term{G |- M':t}$ as follows: $\term{!w:!w |- M\sim M':t}$ if and only if
	\[
		\forall\,W{\sim}W':\type{!w}.~\term{\{W/!w\}M} \sim \term{\{W'/!w\}M'} : \type{t}~.
	\]
	\end{definition}
Machine equivalence validates the equational theory (and in particular the beta and eta equations). Thus we have the following result. 

	\begin{theorem}
	\label{thm:ccc}
	Terms modulo machine equivalence form a Cartesian closed category.
	\end{theorem}
In fact, the category given by terms modulo machine equivalence is just the extensional collapse of the category of terms modulo the equational theory. Note that machine equivalence is strictly coarser than the equational theory: a situation analogous to that of the simply-typed $\lambda$-calculus with products, considered modulo an appropriate contextual equivalence.
%
%


\section{Further work}

The current type system for the FMC is \textit{too strong} for practical programming: it captures such intensional (and unobservable) aspects of computation as the number of elements read from a random stream. We aim to investigate more abstract type systems, including dealing with the particular properties of effectful locations. The results in this paper concerning the type system, which is essentially a presentation of intuitionistic logic, the operational intuition and the close denotational relationship with the $\lambda$-calculus make a strong basis for future refinements which account properly for effects. There are several ways in which we already know how to weaken the type system: introducing a recursor, stream types  $\type{t}^*$, which type a stream of terms of type $\type{t}$, and ignoring types on non-main locations. A close link with string diagrams is evident from the results presented, including with the recently introduced higher-order string diagrams for CCCs \cite{DBLP:journals/corr/abs-2107-13433}. This is another avenue for investigation. 


\bibliographystyle{plainurl}
\bibliography{FMC}


\newpage
\appendix



\section{Details for Section~\ref{sec:SN}: Strong Normalization}
\label{A:SN}

We re-state and prove the lemmata for the strong normalization proof in Section~\ref{sec:SN}.

\begin{theorem}[Lemma~\ref{lem:increasing} restatement]
For all terms $\Gamma \vdash \term{M: t}$ and valuations $v\leq w$ over $\Gamma$, we have that:
\begin{enumerate}
	\item $\snterm{M}_v \in \sntype{t}$ 
	\item $\snterm{M}_v \leq_{\sntype{t}} \snterm{M}_w$.
\end{enumerate}
\end{theorem}

\begin{proof}
We prove both statements simultaneously by induction on the type derivation of $\term{G |- M: t}$.
Recall that the first item is equivalent to claiming $\snterm{M}_v$ is monotonic. When we write `increasing' here, we mean non-strictly.
\begin{itemize}

	\item
For the base case
\begin{align*}
	\Gamma \vdash \term{M} \equiv \term{*: ?tA > !tA},
\end{align*}
\begin{enumerate}
	\item
Observe that $\snterm{*}_v$(t) = (0, t) and so is monotonic.
	\item
Observe that $\snterm{*}_v
  = \snterm{*}_w$
for every $v$ and $w$ over $\Gamma$.
\end{enumerate}
\item
For the abstraction case, where $\term{x: r}$ and
\begin{align*}
	\Gamma \vdash \term{M} \equiv \term{a<x>.M': a(r)\,?sA > !tA},
\end{align*}
\begin{enumerate}

	\item
We must show the function
\begin{align*}
	\snterm{a<x>.M}_v(s, a(r)) = (1+n, t) \textup{ where } (\concat nt) = \snterm{M}_{v\{\term{x} \leftarrow r\}} (s)\
\end{align*} 
is monotonic. Indeed, we have that $(s, a(r)) \leq_{\sntype{!sA\,a(r)}} (s', a(r'))$ implies, by inductive hypothesis (ii) on $\term{M'}$, that
\begin{align*}
 \snterm{M'}_{v\{\term{x} \leftarrow r\}} \leq_{\sntype{?sA > !tA}} \snterm{M'}_{v\{\term{x} \leftarrow r'\}} \ ,
\end{align*}
and then by inductive hypothesis $(i)$ on $\term{M'}$, this implies
\begin{align*}
 \snterm{M'}_{v\{\term{x} \leftarrow r\}}(s) \leq_{\sntype{?sA > !tA}} \sn{\term{M'}}_{v\{\term{x} \leftarrow r'\}}(s') \ .
\end{align*}
Thus, increasing the input of the function increases its output.
\item
We wish to show, for arbitrary $(s, a(r)) \in \sntype{!sA\,a(r)}$, that
\[
v \leq w \quad \textup{implies} \quad \snterm{a<x>.M'}_v(s, a(r)) \leq_{\N \times \sntype{!tA}} \snterm{a<x>.M'}_w(s, a(r))\ .
\]
Unfolding definitions, we see we must show that $v \leq w$ implies
\begin{align*}  (1+n, t) &\leq_{\N \times \sntype{!tA}} (1+n', t')
\\  \textup{where }  (\concat nt) &= \snterm{M}_{v\{\term{x} \leftarrow r\}} (s)\ ,
\\  \textup{and }   (\concat n't') &= \snterm{M}_{w\{\term{x} \leftarrow r\}} (s)\
\end{align*}
By assumption, we have that, for all $r \in \sntype{r}$, $v\{\term{x} \leftarrow r\} \leq w\{\term{x} \leftarrow r\}$. Applying inductive hypothesis $(ii)$ on $\term{M'}$, we thus have that
\[
	\snterm{M'}_{v\{\term{x} \leftarrow r\}} \leq_{\sntype{?sA > !tA}} \snterm{M'}_{w\{\term{x} \leftarrow r\}}\ .
\]
and consequently $(n, t) \leq_{\N \times \sntype{!tA}} (n', t')$.
Indeed, the required result
immediately follows.
\end{enumerate}

	\item
For the application case, with $\Gamma \vdash \term{N: r}$ and
\begin{align*}
	\Gamma \vdash \term{M} \equiv \term{[N]a.M': ?sA > !tA},
\end{align*}
\begin{enumerate}
\item
We have to show that
\begin{align*}
 	\snterm{[N]a.M'}_v (s)  &= (1+n + \collapse{\snterm{N}_v}, t)
	\\		 \textup{where } (n, t) &= \snterm{M'}_v (s, a(\snterm{N}_v))\ .
\end{align*}
is monotonic. Applying inductive hypothesis $(i)$ on $\term{M'}$, we achieve that  $\snterm{M'}_v$ is monotonic. Thus, increasing the input $s$ increases $(n, t)$, which therefore increases the output of the entire function.
\item
We have to show, for arbitrary $s \in \sntype{!sA}$, that
\[
v \leq w \qquad \textup{implies} \qquad \snterm{[N]a.M'}_v (s) \leq_{\N \times \sntype{!tA}} \snterm{[N]a.M'}_w (s)
\]
Unfolding definitions, we see we must show that $v \leq w$ implies
\begin{align*}
 	(1+n + \collapse{\snterm{N}_v}, t) & \leq_{\N \times \sntype{!tA}} (1+n' + \collapse{\snterm{N}_w}, t')
	\\		 \textup{where } (n, t) &= \snterm{M'}_w (s, a(\snterm{N}_v))\ .
	\\		 \textup{and } (n', t') &= \snterm{M'}_w (s, a(\snterm{N}_w))\ .
\end{align*}
Applying inductive hypothesis $(ii)$ on $\term{M'}$ and $\term{N}$, we achieve
\[
	\sn{\term{M'}}_v \leq_{\sntype{a(r)\,?sA > !tA}} \sn{\term{M'}}_w \quad \textup{and} \quad \sn{\term{N}}_v \leq_{\sntype{r}} \sn{\term{N}}_w.
\]
The conjunction of both statements implies that
$(n, t) \leq_{\N \times \sntype{!t}} (n', t')$ .
We aditionaly observe that $\collapse{\snterm{N}_v} \leq_{\N} \collapse{\snterm{N}_w}$, and the result follows.
\end{enumerate}
\item
For the variable case,
\begin{align*}
	\Gamma, \term{x: ?rA > !uA} \vdash \term{M} \equiv \term{x.M': ?rA?sA > !tA},
\end{align*}
\begin{enumerate}
\item
We have to show that
\begin{align*}
 \snterm{x.M'}_v(\concat sr) &= (\concat{n+m}t)
\\       \textup{where } (\concat mt) {}&=\snterm{M'}_v (\concat su)
    \\   \textup{and } (\concat nu) {}&= v(\term{x})(r)\ .
\end{align*}
is monotonic.
Observe that $v(\term{x}) \in \sn{\type{?rA > !uA}}$ and so is monotonic. Thus, increasing the input $(s, r)$ increases $(n, u)$. We have from  inductive hypothesis (i) on $\term{M'}$ that  $\sn{ \term{M'}}$ is monotonic.  Altogether, this results in an increase in $(m, t)$ , which therefore increases the output of the entire function.

	\item
We wish to show, for arbitrary $(s, r) \in \type{!sA!rA}$, that
\[
	v \leq w \qquad \textup{implies} \qquad \snterm{x.M'}_v(s,r) \leq_{\N \times \sntype{!tA}} \snterm{x.M'}_w(s,r)\ .
\]
Unfolding definitions, we see we must show that $v \leq w$ implies
\begin{align*}
 (\concat{n+m}t) &\leq_{\N \times \sntype{!tA}} (\concat{n'+m'}t')
\\       \textup{where } (\concat mt) {}&=\snterm{M'}_v (\concat su)
    \\   \textup{and   }\ \  (\concat nu) {}&= v(\term{x})(r)
\\       \textup{and } (\concat m't') {}&=\snterm{M'}_w (\concat su')
    \\   \textup{and } (\concat n'u') {}&= w(\term{x})(r)\ .
\end{align*}
By assumption, we have $v(\term{x}) \leq_{\sntype{?rA > !uA}} w(\term{x})$, which implies that
$(n, u) \leq_{\N \times \sntype{!uA}} (n', u')$. Applying inductive hypothesis (ii) on $\term{M'}$, we have that
\[
	\sn{\term{M'}}_v \leq_{\sntype{?uA?sA>!tA}}  \sn{\term{M'}}_w.
\]
Altogether, this implies $(m, t) \leq_{\N \times \sntype{!tA}} (m', t')$.
Thus, we achieve the required result.
\qedhere
\end{enumerate}
\end{itemize}
\end{proof}


We add the following \emph{Weakening Lemma}.

\begin{lemma}[Weakening]
\label{lem:weakening}
For all terms $\term{G |- M: t}$, valuations $v$ on $\Gamma$, and $s \in \sntype{s}$, we have that
\begin{align*}
	\snterm{G |- M:t}_{v} = \snterm{G , x:s |- M:t}_{v\{\term{x} \leftarrow s\}}
\end{align*}
where $\term{x} \notin \fv{\term{M}}$.
\end{lemma}
\begin{proof}
Induction on the type derivation of $\term{G |- M: t}$.
\end{proof}


\begin{theorem}[Lemma~\ref{lem:sequencing} restatement]
For terms $\term{G |- M: ?sA?tA > !uA}$ and $\term{G |- N: ?rA > !sA}$ and valuation $v$ on $\Gamma$,
\begin{align*}
	\snterm{N;M}_v(t,r) = (i+j,u)
	\quad\text{where}\quad	\snterm{N}_v  (r)=(i,s) 
	\quad\text{and}  \quad  \snterm{M}_v(t,s)=(j,u)~.
\end{align*}
\end{theorem}

\begin{proof}
We proceed by induction on the type derivation of $\term{G |- N:t}$.
\begin{itemize}

	\item
Case $\term{G |- * : ?sA > !sA}$. Given $s\in\sntype{!sA}$ and $t\in\sntype{!tA}$, let $\snterm{M}_v(t,s)=(m,u)$. Since $\snterm{*}_v(s)= (0,s)$, we need to show that $\snterm{*;M}_v(t,s)=(0+m,u)$, but this is immediate since $\term{*;M}=\term M$.

	\item
Case $\term{G , x: ?nA > !pA |- x.N : ?nA\,?rA > !sA}$ where $\term{G ,x: ?nA > !pA |- N : ?pA\,?rA > !sA}$. Given $n\in\sntype{!nA}$, $r\in\sntype{!rA}$ and $t\in\sntype{!tA}$, let 
\[
	v(\term x)(n)=(i,p) \qquad \snterm N_v(r,p) = (j,s) \qquad \snterm M_v(t,s) = (k,u)
\]
so that $\snterm{x.N}_v(r,n)=(i+j,s)$. We need to show that $\snterm{x.N;M}_v(t,r,n)=(i+j+k,u)$. By definition, $\term{x.N;M}=\term{x.(N;M)}$. The inductive hypothesis gives $\snterm{N;M}_v(t,r,p)=(j+k,u)$, so that $\snterm{x.(N;M)}_v(t,r,n)=(i+j+k,u)$ as required.

	\item
Case $\term{G |- [P]a.N : ?rA > !sA}$ where $\term{G |- P : p}$ and $\term{G |- N : a(p)\,?rA > !sA}$. Given $r\in\sntype{!rA}$ and $t\in\sntype{!tA}$, let 
\[
	\snterm P_v=p \qquad \snterm N_v(r,p)=(i,s) \qquad \snterm M_v(t,s) = (j,u)
\]
so that $\snterm{[P]a.N}_v(r)=(\floor p+1+i,s)$. We need to show that $\snterm{[P]a.N;M}_v(t,r)=(\floor p+1+i+j)$. By definition, $\term{[P]a.N;M}=\term{[P]a.(N;M)}$. The inductive hypothesis gives $\snterm{N;M}_v(t,r,p)=(i+j,u)$, so that $\snterm{[P]a.(N;M)}_v(t,r)=(\floor p+1+i+j)$, as required.

	\item
Case $\term{G |- a<x>.N : a(p)\,?rA > !sA}$ where $\term{G , x:p |- N: ?rA > !sA}$. Given $p\in\sntype{p}$, $r\in\sntype{!rA}$ and $t\in\sntype{!tA}$, let
\[
	\snterm{N}_{v\{\term x\from p\}}(r)=(i,s) \qquad \snterm{M}_v(t,s)=(j,u)
\]
so that $\snterm{a<x>.N}_v(r,p)=(1+i,s)$. We need to show that $\snterm{a<x>.N;M}_v(t,r,p)=(1+i+j,u)$. By definition, $\term{a<x>.N;M}=\term{a<x>.(N;M)}$. We assume $\term x$ is not free in $\term M$ (otherwise, $\alpha$-rename it in $\term{a<x>.N}$); then by Lemma~\ref{lem:weakening} also $\snterm{M}_{v\{\term x\from p\}}(t,s)=(j,u)$. The inductive hypothesis gives $\snterm{N;M}_{v\term x\from p\}}(t,r)=(i+j,u)$, so that $\snterm{a<x>.(N;M)}_v(t,r,p)=(1+i+j,u)$, as required.
\qedhere
\end{itemize}
\end{proof}


\begin{theorem}[Lemma~\ref{lem:substitution} restatement]
For terms $\term{G |- N:s}$ and $\term{G , x:s |- M:t}$ and valuation $v$ on $\Gamma$,
\[
	\snterm{\{N/x\}M}_v = \snterm M_{v\{\term x \from \snterm N_v\}}~.
\]
\end{theorem}

\begin{proof}
We proceed by induction on the type derivation of $\Gamma, \term{x:w} \vdash \term{M: t}$.
\begin{itemize}
\item For the base case,
\begin{align*}
	\Gamma, \term{x:w} \vdash \term{M} \equiv \term{*: ?t_A > !t_A},
\end{align*}
observe that $\snterm{D |- *}_v$ is independent of $v$ and of $\Delta$. Thus, we have
\begin{align*}
\snterm{G |- \{N/x\}*}_v (t) = \snterm{G |- *}_v(t) = (0,t) = \snterm{G, x |- *}_{v\{\term x \leftarrow \snterm{N}_v\}}(t)\ ,
\end{align*}
as required.
\item
For the abstraction case,
\begin{align*}
	\Gamma, \term{x:w} \vdash \term{M} \equiv \term{a<y>.M': `a(r)?s_A > !t_A},
\end{align*}
we need to show that
\begin{align*}
\snterm{\{N/x\}a<y>.M'}_v(s, a(r)) = \snterm{a<y>.M'}_{v\{\term x \leftarrow \snterm{N}_v\}}(s, a(r)) \ ,
\end{align*}
where $\term x \neq \term y$. For the left-hand side, we have
\begin{align*}
 \sn{\term{\{N/x\}a<y>.M'}}_{v} ({s}, a(r)) &= \\
\sn{\term{a<y>.\{N/x\}M'}}_{v} ({s}, a(r)) & = (1+n, {t})\\
\textup{where } (n, {t}) &= \sn{\term{\{N/x\}M'}}_{v\{\term{y} \leftarrow r\}} ({s})\ ,
\end{align*}
and, for the right-hand side, we have
\begin{align*}
\sn{\term{a<y>.M'}}_{v\{\term x \leftarrow \snterm{N}_v\}} ({s}, a(r)) &= (1+n', {t'}) \\
\textup{where } (n', {t'}) &= \sn{\term{M'}}_{u\{\term x \leftarrow \snterm{N}_v\}} ({s}) \ .
\end{align*}
Let $u = v\{\term{y} \leftarrow r\}$.
Applying the inductive hypothesis on $\term{M'}$ gives the first equality below.
\[
	\snterm{\{N/x\}M}_{v\{\term y \leftarrow r\}} = \snterm{M'}_{u\{\term x \leftarrow \snterm{N}_u\}} =  \snterm{M'}_{u\{\term x \leftarrow \snterm{N}_v\}}
\]
By Lemma~\ref{lem:weakening} (the Weakening Lemma),
and since $\term y \not\in \fv{\term{N}}$, we have $\snterm{N}_u = \snterm{N}_v$. This gives the second equality.
Thus, we have that $(n, t) = (n', t')$ and the required result follows.
\item
For the application case, where $\Gamma, \term{x: c} \vdash  \term{P: r}$ and
\begin{align*}
	\Gamma, \term{x:w} \vdash \term{M} \equiv \term{[P]a.M': ?s_A > !t_A},
\end{align*}
we need to show that
\[
	\snterm{\{N/x\}[P]a.M'}_v(s) = \snterm{[P]a.M'}_{v\{\term x \leftarrow \snterm{N}_v\}}(s)\ .
\]
For the left-hand side, we have
\begin{align*}
 \sn{\term{\{N/x\}[P]a.M'}}_{v}({s}) &= \\
 \sn{\term{[\{N/x\}P]a.\{N/x\}M'}}_{v}({s}) &= (1+n+\collapse{\snterm{P\{N/x\}}_v}, t)  \\
 \textup{where } (n, t) &=  \sn{\term{\{N/x\}M'}}_{v}  ({s}, a(\sn{\term{\{N/x\}P}}_{v})) \ ,
\end{align*}
and, for the right-hand side, we have
\begin{align*}
\sn{\term{[P]a.M'}}_{v\{\term x \leftarrow \snterm{N}_v\}}({s}) &= (1+n'+\collapse{\snterm{P}_{v\{\term x \leftarrow \snterm{N}_v\}}}, t') \\
 \textup{where } (n', t') &= \sn{\term{M'}}_{v\{\term x \leftarrow \snterm{N}_v\}}  ({s}, a(\sn{ \term{P}}_{v\{\term x \leftarrow \snterm{N}_v\}})) \ .
\end{align*}
Applying the inductive hypothesis on  $\term{M'}$ and $\term{P}$ achieves
\begin{align*}
 \snterm{\{N/x\}M}_{v} &= \snterm{M'}_{v\{\term x \leftarrow \snterm{N}_v\}} \\
 \snterm{\{N/x\}P}_{v} &= \snterm{P}_{v\{\term x \leftarrow \snterm{N}_v\}} \  .
\end{align*}
Thus, we have that $(n, t) = (n', t')$ and indeed $1+n'+\collapse{\snterm{\{N/x\}P}_{v}} = 1+n'+\collapse{\snterm{P}_{v\{\term x \leftarrow \snterm{N}_v\}}}$ as required.
\item
For the variable case, where $\term{y} \neq \term{x}$,
\begin{align*}
	\Gamma, \term{x:w}, \term{y: ?r_A > !u_A} \vdash \term{M} \equiv \term{y.M': ?r_A?s_A> !t_A},
\end{align*}
 we need to show that
\[
	\snterm{\{N/x\}y.M}_v(s, r) = \snterm{y.M}_{v\{\term x \leftarrow \snterm{N}_v\}}(s,r)\ .
\]
For the left-hand side,
\begin{align*}
 \sn{\term{\{N/x\}y.M'}}_{v}({s}, {r}) &= \\
 \sn{\term{y.\{N/x\}M'}}_{v}({s}, {r}) &= (n+m, {t}) \\
  \textup{where } (m, {t}) &= \sn{\term{\{N/x\}M'}}_{v} ({s}, {u}) \\
  \textup{and }  (n, {u}) &= v(\term{y})({r})\ ,
\end{align*}
and for the right-hand side,
\begin{align*}
\sn{\term{y.M'}}_{v\{\term x \leftarrow \snterm{N}_v\}} ({s}, {r}) &= (n'+m', {t'}) \\
  \textup{where } (m', {t'}) &= \sn{\term{M'}}_{v\{\term x \leftarrow \snterm{N}_v\}} ({s}, {u'}) \\
   \textup{and } (n', {u'}) &= v\{\term x \leftarrow \snterm{N}_v\}(\term{y})({r})  \ .
\end{align*}
Observe that $v\{\term x \leftarrow \snterm{N}_v\}(\term{y}) = v(\term{y})$, which implies  $(n, u) = (n', u')$.
Application of the inductive hypothesis on $\term{M'}$ achieves
\begin{align*}
 \snterm{\{N/x\}M}_{v} &= \snterm{M'}_{v\{\term x \leftarrow \snterm{N}_v\}} \ .
\end{align*}
Thus, we have $(m, t) = (m', t')$ and the result follows.
\item
For the variable case,
\begin{align*}
	\Gamma, \term{x:?r_A > !u_A} \vdash \term{M} \equiv \term{x.M': ?r_A?s_A > !t_A},
\end{align*}
we need to show that
\[
	\snterm{\{N/x\}x.M}_v(s, r) = \snterm{x.M}_{v\{\term x \leftarrow \snterm{N}_v\}}(s,r)\ .
\]
For the left-hand side, by application of the Sequencing Lemma \ref{lem:sequencing}, we have
\begin{align*}
 \sn{\term{\{N/x\}x.M'}}_{v}  ({s}, {r}) &=\\
 \sn{\term{N;\{N/x\}M'}}_{v}  ({s}, {r}) &= (n+m, {t}) \\
    \textup{where } (m, {t}) &= \sn{ \term{\{N/x\}M'}}_{v} ({s}, {u}) \\
    \textup{and } (n, {u}) &=  \sn{\term{N}}_v ({r})  \ ,
\end{align*}
and, for the right-hand side,
\begin{align*}
 \sn{ \term{x.M'}}_{v\{\term x \leftarrow \snterm{N}_v\}} ({s}, {r}) &= (n'+m', {t'})  \\
    \textup{where }   (m', {t'}) &= \sn{\term{M'}}_{v\{\term x \leftarrow \snterm{N}_v\}} ({s}, {u'})\\
  \textup{and } (n', {u'}) &= v\{\term x \leftarrow \snterm{N}_v\}(\term{x})({r}) \ .
\end{align*}
Observe that  $v\{\term x \leftarrow \snterm{N}_v\}(\term{x}) = \sn{\term{N}}_{v}$, which implies $(n, u) = (n', u')$.
Application of the inductive hypothesis on $\term{M'}$ achieves
\begin{align*}
 \snterm{\{N/x\}M}_{v} &= \snterm{M'}_{v\{\term x \leftarrow \snterm{N}_v\}} \ .
\end{align*}
Thus, we have $(m, t) = (m', t')$ and the result follows.
\qedhere
\end{itemize}
\end{proof}


We simplify the remaining proof by separating the $\beta$-rule
\[
	\term{[N]a.H.a<x>.M}~\rw~\term{H.\{N/x\}M}
\]
into a \emph{permutation} part 
\[
	\term{[N]a.H.a<x>.M}~\sim_p~\term{H.[N]a.a<x>.M}
\]
and a \emph{strict reduction} part
\[
	\term{[N]a.a<x>.M}~\rw~\term{\{N/x\}M}~.
\]
The \emph{permutation equivalence} $\sim$ on terms is given by:
\[
\begin{aligned}
	\term{[P]a.[N]b.M} &~\sim_p~ \term{[N]b.[P]a.M}
\\	\term{a<x>.[N]b.M} &~\sim_p~ \term{[N]b.a<x>.M} & \text{if }x\notin\fv{\term N}
\\	\term{a<x>.b<y>.M} &~\sim_p~ \term{b<y>.a<x>.M}
\end{aligned}
\]
where in each case $a\neq b$. We first show that the interpretation of terms is preserved under permutations (the \emph{Permutation Lemma}).


\begin{lemma}[Permutation]
\label{lem:permutation}
For $\term{G |- M `{\,\black\sim_p\,} N : t}$ and any valuation $v$ on $\Gamma$, $\snterm M_v = \snterm N_v$.
\end{lemma}

\begin{proof}
The most interesting case is
\begin{align*}
	\term{a<x>.[N]b.M} \sim_p \term{[N]b.a<x>.M: `a(r)?s_A > !t_A} \ ,
\end{align*}
where $\term{x} \not\in \fv{\term{N}}$.
To see this, observe the following are equivalent:
\begin{align*}
 \sn{\term{a<x>.[N]b.M: }}_v ({s}, a(r)) &= (n+2+\collapse{\snterm{N}_{v\{\term{x} \leftarrow r\}}}, {t}), \\
 \textup{where } (\concat{n}{{t}}) &= \sn{\term{M}}_{v\{\term{x} \leftarrow r\}} ({s}, b(\sn{\term{N}}_{v\{\term{x} \leftarrow r\}}), \\
 \sn{\term{[N]b.a<x>.M}}_v ({s}, a(r)) &= (n+2+\collapse{\snterm{N}_{v}}, {t}), \\
   \textup{where } (\concat{n}{{t}}) &= \sn{\term{P}}_{v\{\term{x} \leftarrow r\}}  ({s}, b(\sn{\term{N}}_{v} ),
\end{align*}
using Lemma~\ref{lem:weakening}, the Weakening Lemma (since $\term{x} \not\in \fv{\term{N}}$). The other cases follow immediately from unfolding definitions.
\end{proof}


\begin{theorem}[Lemma~\ref{lem:reduction monotone} restatement]
If $\term{G |- M -> N : t}$ then $\snterm{M}_v \geq_{\sntype{t}} \snterm{N}_v$ for every valuation v on $\Gamma$.
\end{theorem}

\begin{proof}
We proceed by induction on the derivation of one-step reductions. Using the Permutation Lemma \ref{lem:permutation}, we need only consider the \textit{strict} $\beta$-rule in the base case.
\begin{itemize}
\item The base case,
\begin{align*}
	\term{[N]a.a<x>.M}~\rw~\term{\{N/x\}M} : \type{?s_A > !t_A},
\end{align*}
where $\term{x, N: r}$ requires showing, for arbitrary $s \in \sntype{!s_A}$, that
\begin{align*}
	\snterm{[N]a.a<x>.M}_v(s)~\geq_{\N \times \sntype{!t_A}}~\snterm{\{N/x\}M}_v(s) \ .
\end{align*}
Unfolding definitions, we must show
\begin{align*}
(2 + n + \collapse{\snterm{N}_{v}}, t) & \geq_{\N \times \sntype{!t_A}}~\snterm{\{N/x\}M}_v(s)
\\		 \textup{where } (n, t) &= \snterm{M}_{v\{x \leftarrow \snterm{N}_v \}}(s)\ .
\end{align*}
Applying the Substitution Lemma \ref{lem:substitution},
\[
	\snterm{M}_{v\{x \leftarrow \snterm{N}_v \}} =  \snterm{\{N/x\}M}_{v},
\]
 we see $(n, t) = \snterm{\{N/x\}M}_{v}(s)$.
Observing that $(2 + n + \collapse{\snterm{N}_{v}}, t) \geq_{\N \times \sntype{!t_A}} (n, t)$, the result follows.
\item The abstraction case,
\begin{align*}
	\term{a<x>.M}~\rw ~\term{a<x>.M'} : \type{ `a(r)?s_A > !t_A},
\end{align*}
requires showing, for arbitrary $(s, a(r)) \in \sntype{!s_A`a(r)}$, that
\[
	\snterm{a<x>.M}_v(s, a(r))~ \geq_{\N \times \sntype{!t_A}} ~\snterm{a<x>.M'}_v(s, a(r)) \ .
\]
Unfolding definitions, we must show
\begin{align*}
	(1+n, t) &\geq_{\N \times \sntype{!t_A}} (1+n', t')
\\  \textup{where } (\concat nt) &= \snterm{M}_{v\{\term{x} \leftarrow r\}} (s)
\\  \textup{and } (\concat {n'}t') &= \snterm{M'}_{v\{\term{x} \leftarrow r\}} (s)\ .
\end{align*}
Applying the inductive hypothesis on $\term{M} \rw \term{M'}$, we achieve
\[
	\snterm{M}_w(s) \geq_{\N \times \sntype{!t_A}} \snterm{M'}_{w}(s)  \ ,
\]
for any valuation $w$. In particular, we can set $w = \{\term{x} \leftarrow r\}$
and thus we have that $(n, t) \geq_{\N \times \sntype{!t_A}} (n', t')$. The result follows.
\item The application, function case,
\begin{align*}
	\term{[N]a.M}~\rw ~\term{[N]a.M'} : \type{ ?s_A > !t_A},
\end{align*}
requires showing, for arbitrary $s \in \sntype{!s_A}$, that
\[
	\snterm{[N]a.M}_v(s)~ \geq_{\N \times \sntype{!t_A}} ~\snterm{[N]a.M'}_v(s) \ .
\]
Unfolding definitions, we must show
\begin{align*}
	(1+n+\collapse{\snterm{N}_v}, t) &\geq_{\N \times \sntype{!t_A}} (1+n'+\collapse{\snterm{N}_v}, t')
\\  \textup{where } (\concat nt) &= \snterm{M}_{v}(s, a(\snterm{N}_v))
\\  \textup{and } (\concat n't') &= \snterm{M'}_{v}(s, a(\snterm{N}_v)) \ .
\end{align*}
Applying the inductive hypothesis on   $\term{M} \rw \term{M'}$,  we achieve
\[
	\snterm{M}_w(s, a(\snterm{N}_v)) \geq_{\N \times \sntype{!t_A}} \snterm{M'}_{w}(s, a(\snterm{N}_v))  \ ,
\]
and thus that $(n, t) \geq_{\N \times \sntype{!t}} (n', t')$. The result follows.
\item The application, argument case, where $\term{N: r}$,
\begin{align*}
	\term{[N].M}~\rw ~\term{[N'].M} : \type{ ?s_A > !t_A},
\end{align*}
requires showing, for arbitrary $s \in \sntype{!s_A}$, that
\[
	\snterm{[N]a.M}_v(s)~ \geq_{\N \times \sntype{!t_A}} ~\snterm{[N']a.M}_v(s) \ .
\]
Unfolding definitions, we must show
\begin{align*}
(1+n+\collapse{\snterm{N}_v}, t) &\geq_{\N \times \sntype{!t_A}} (1+n'+\collapse{\snterm{N'}_v}, t')
\\  \textup{where } (\concat nt) &= \snterm{M}_{v}(s, a(\snterm{N}_v))
\\  \textup{and } (\concat n't') &= \snterm{M}_{v}(s, a(\snterm{N'}_v))
\end{align*}
Applying the inductive hypothesis on $\term{N} \rw \term{N'}$, we achieve
\[
	\snterm{N}_v \geq_{\sntype{r}} \snterm{N'}_v
\]
which allows us to apply monotonicity of $\snterm{M}_v$ to see that $(n, t) \geq_{\N \times \sntype{!t_A}} (n', t')$. It then follows that $ (1+n+\collapse{\snterm{N}_v}, t)\geq_{\N \times \sntype{!t_A}}  (1+n'+\collapse{\snterm{N'}_v}, t')$, as required.
\item The variable case,
\begin{align*}
	\term{x.M}~\rw~ \term{x.M'} : \type{?r_A?s_A > !t_A},
\end{align*}
with $\term{x: ?r_A > !u_A}$
requires showing, for arbitrary $(s,r) \in \sntype{!s_A!r_A}$, that
\[
	\snterm{x.M}_v(s, r)~ \geq_{\N \times \sntype{!t_A}} ~\snterm{x.M'}_v(s, r) \ .
\]
Unfolding definitions, we must show
\begin{align*}
	 (n+m, t) &\geq_{\N \times \sntype{!t_A}} (n+m', t') \\
  \textup{where } (\concat mt) &= \snterm{M}_{v}(s, u) \\
  \textup{and } (\concat m't') &= \snterm{M'}_{v}(s, u) \\
  \textup{and } (n, u) &= v(\term{x})(r)
\end{align*}
Applying the inductive hypothesis on   $\term{M} \rw \term{M'}$,  we achieve, for any $u \in \sntype{!u_A}$,
\[
	\snterm{M}_w(s, u) \geq_{\N \times \sntype{!t_A}} \snterm{M'}_{w}(s, u)  \ ,
\]
and thus that $(m, t) \geq_{\N \times \sntype{!t_A}} (m', t')$ and consequently $(n+m, t) \geq_{\N \times \sntype{!t_A}} (n+m', t')$, as required.
\qedhere
\end{itemize}
\end{proof}


\begin{theorem}[Lemma~\ref{lem:reduction collapsed} restatement]
If $\term{G |- M -> N: ?sA > !tA}$ then $\pi_1(\snterm M_v(s)) >_\N \pi_1(\snterm N_v(s))$ for every $s\in\sntype{!sA}$ and valuation $v$ on $\Gamma$.
\end{theorem}

\begin{proof}
We proceed by induction on the derivation of one-step reductions. Using the Permutation Lemma \ref{lem:permutation}, we need only consider the \textit{strict} $\beta$-rule in the base case.
\begin{itemize}

	\item
The base case,
\begin{align*}
	\term{[N]a.a<x>.M}~\rw~\term{\{N/x\}M} : \type{?s_A > !t_A},
\end{align*}
where $\term{x, N: r}$ requires showing, for arbitrary $s \in \sntype{!s_A}$, that
\begin{align*}
	\pi_1(\snterm{[N]a.a<x>.M}_v(s))~ >_{\N}~ \pi_1(\snterm{\{N/x\}M}_v(s)) \ .
\end{align*}
Unfolding definitions, we must show
\begin{align*}
2 + n + \collapse{\snterm{N}_{v}} & >_{\N} ~\pi_1(\snterm{\{N/x\}M}_v(s))
\\		 \textup{where } (n, t) &= \snterm{M}_{v\{x \leftarrow \snterm{N}_v \}}(s)\ .
\end{align*}
Applying the Substitution Lemma \ref{lem:substitution},
\[
	\snterm{M}_{v\{x \leftarrow \snterm{N}_v \}} =  \snterm{\{N/x\}M}_{v},
\]
 we see $(n, t) = \snterm{\{N/x\}M}_{v}(s)$.
Indeed, the required result follows immediately from observing that  $n = \pi_1( \snterm{\{N/x\}M}_{v}(s))$.
\item The abstraction case,
\begin{align*}
	\term{a<x>.M}~\rw ~\term{a<x>.M'} : \type{ `a(r)?s_A > !t_A},
\end{align*}
requires showing, for arbitrary $(s, a(r)) \in \sntype{!s_A`a(r)}$, that
\[
	\pi_1(\snterm{a<x>.M}_v(s, a(r)))~ >_{\N} ~\pi_1(\snterm{a<x>.M'}_v(s, a(r))) \ .
\]
Unfolding definitions, we must show
\begin{align*}
	1+n &>_{\N} 1+n'
\\  \textup{where } (\concat nt) &= \snterm{M}_{v\{\term{x} \leftarrow r\}} (s)
\\  \textup{and } (\concat {n'}t') &= \snterm{M'}_{v\{\term{x} \leftarrow r\}} (s)\ .
\end{align*}
Applying the inductive hypothesis on $\term{M} \rw \term{M'}$, we achieve
\[
	\pi_1(\snterm{M}_w(s)) >_{\N} \pi_1(\snterm{M'}_{w}(s))  \ ,
\]
for any valuation $w$. In particular, we can set $w = \{\term{x} \leftarrow r\}$ and thus we have that $n >_{\N} n'$, as required.
\item The application, function case,
\begin{align*}
	\term{[N]a.M}~\rw ~\term{[N]a.M'} : \type{ ?s_A > !t_A},
\end{align*}
requires showing, for arbitrary $s \in \sntype{!s_A}$, that
\[
	\pi_1(\snterm{[N]a.M}_v(s))~ >_{\N} ~\pi_1(\snterm{[N]a.M'}_v(s)) \ .
\]
Unfolding definitions, we must show
\begin{align*}
	1+n+\collapse{\snterm{N}_v} &>_{\N} 1+n'+\collapse{\snterm{N}_v}
\\  \textup{where } (\concat nt) &= \snterm{M}_{v}(s, a(\snterm{N}_v))
\\  \textup{and } (\concat n't') &= \snterm{M'}_{v}(s, a(\snterm{N}_v)) \ .
\end{align*}
Applying the inductive hypothesis on   $\term{M} \rw \term{M'}$,  we achieve
\[
	\pi_1(\snterm{M}_w(s, a(\snterm{N}_v)))  >_{\N} \pi_1(\snterm{M'}_{w}(s, a(\snterm{N}_v)))  \ ,
\]
and thus that $n >_{\N} n'$, as required.
\item The application, argument case, where $\term{N: r}$,
\begin{align*}
	\term{[N].M}~\rw ~\term{[N'].M} : \type{ ?s_A > !t_A},
\end{align*}
requires showing, for arbitrary $s \in \sntype{!s_A}$, that
\[
	\pi_1(\snterm{[N]a.M}_v(s))~ >_{\N} ~\pi_1(\snterm{[N']a.M}_v(s)) \ .
\]
Unfolding definitions, we must show
\begin{align*}
1+n+\collapse{\snterm{N}_v} &>_{\N} 1+n'+\collapse{\snterm{N'}_v}
\\  \textup{where } (\concat nt) &= \snterm{M}_{v}(s, a(\snterm{N}_v))
\\  \textup{and } (\concat n't') &= \snterm{M}_{v}(s, a(\snterm{N'}_v))
\end{align*}
Applying the inductive hypothesis on $\term{N} \rw \term{N'}$, we achieve
\[
	\pi_1(\snterm{N}_v(s)) >_{\N} \pi_1(\snterm{N'}_v(s))
\]
In the special case where $s$ is the minimal element, we recover $\collapse{\snterm{N}_v} >_{\N} \collapse{\snterm{N'}_v}$.
Additionally, by Lemma \ref{lem:reduction monotone}, we have that $\snterm{N} \geq_{\N \times \sntype{!t}} \snterm{N'}$. This allows us to apply monotonicity of $\snterm{M}$ to deduce that $(n, t) \geq_{\N \times \sntype{!t}} (n', t')$.
Altogether, this suffices for the result.
%
\item
The variable case,
\begin{align*}
	\term{x.M}~\rw~ \term{x.M'} : \type{?r_A?s_A > !t_A},
\end{align*}
with $\term{x: ?r_A > !u_A}$
requires showing, for arbitrary $(s,r) \in \sntype{!s_A!r_A}$, that
\[
	\pi_1(\snterm{x.M}_v(s, r))~ >_{\N} ~\pi_1(\snterm{x.M'}_v(s, r)) \ .
\]
Unfolding definitions, we must show
\begin{align*}
	 n+m &>_{\N} n+m' \\
  \textup{where } (\concat mt) &= \snterm{M}_{v}(s, u) \\
  \textup{and } (\concat m't') &= \snterm{M'}_{v}(s, u) \\
  \textup{and } (n, u) &= v(\term{x})(r)
\end{align*}
Applying the inductive hypothesis on   $\term{M} \rw \term{M'}$,  we achieve, for any $u \in \sntype{!u_A}$,
\[
	\pi_1(\snterm{M}_w(s, u)) >_{\N} \pi_1(\snterm{M'}_{w}(s, u))  \ ,
\]
and thus that $m >_{\N} m'$, as required.
\qedhere
\end{itemize}
\end{proof}


\section{Proofs for Section~\ref{sec:categories}: Categorical Semantics}
\label{A:semantics}

We formally define the category $\SCatSim$ of sequential $\lambda$-terms generated by a set of base types $\Sigma$, modulo an equivalence $\sim$ as follows.
\begin{definition}
The category $\SCatSim$ is given by: \emph{objects} are type vectors $\type{!t}$, and a \emph{morphism} from $\cat{!s}$ to $\cat{!t}$ is an equivalence class under $\sim$ of closed sequential $\lambda$-terms $\term{N:?s>!t}$. (Diagrammatic) composition is sequencing $\term{N;M}$ with identity $\id=\term*$. Open terms are included in the category by closing them:
\[
	\term{?x:?r |- M:?s>!t}\quad\mapsto\quad\term{|- <?x>.M:?r\,?s>!t}~.
\]
\end{definition}
We will vary the equivalence $\sim$ throughout this section.

\subsection{A pre-monoidal category}

\begin{theorem}[Proposition~\ref{prop:premonoidal} restatement]
Terms modulo $\beta\eta$-equivalence form a strict premonoidal category.
\end{theorem}

\begin{proof}
The associativity and unitality morphisms are identities on type vectors. It remains to show that the left action $\cat{- \ten !t}$ and right action $\cat{!t \ten -}$ are functorial. Let $\term{M:?r>!s}$ and $\term{N:?s>!u}$, and let $\term{!x:!t}$ and $\term{!y:!t}$. For the left action:
\[
	\id_{\cat{!s}}\cat{*!t}~=~\term{*:?s?t>!t!s}~=~\id_{\cat{!s*!t}}
\]
\[
\begin{aligned}
	&\cat{(M*!t);(N*!t)}
\\	&\quad=~(\term{<?x>.M.[!x]:?t?r>!s!t})\,;\,(\term{<?y>.N.[!y]:?t?s>!u!t})
\\	&\quad=~(\term{<?x>.M.[!x].<?y>.N.[!y]:?t?r>!u!t})
\\	&\quad=_\beta\;(\term{<?x>.M.N.[!x]:?t?r>!u!t})
\\	&\quad=~\cat{(M;N)*!t}\ .
\end{aligned}
\]
For the right action:
\[
	\cat{!t*}\id_{\cat{!s}}
	~=~\term{<?x>.[!x]:?t?s>!s!t}
	~=_\eta~\term{*:?t?s>!s!t}
    ~=~\id_{\cat{!t*!s}}
\]
\[
\begin{aligned}
	&\cat{(!t*M);(!t*N)}
\\	&\quad=~(\term{M:?r?t>!t!s})\,;\,(\term{N:?s?t>!t!u})
\\	&\quad=~\term{M;N:?r?t>!t!u}
\\	&\quad=~\cat{!t*(M;N)}\ .
\end{aligned}\qedhere
\]
\end{proof}

\subsection{A Cartesian closed category}
We now take $\sim$ to be the equational theory defined in Section 6 and show that $\SCatSim$ is in fact a Cartesian closed category. 
\begin{remark}
For the $\eta$-law, taking $\term{P} = \term{*: (?r > !t) > (?r > !t)}$ in $\term * = \term{<x>.[[x].P.<z>.z]}$ results in $\term{*} = \term{<x>.[[x].<z>.z]} = \term{<x>.[x]}$ at higher type, so $\term{*} = \term{<x>.[x]}$ holds for all types. Iterating results in $\term{*} = \term{<?x>.[!x]}$. We will freely use this as a case of the $\eta$-law without mention.
\end{remark}
\begin{theorem}[Theorem \ref{thm:sound-ccc-eqns} restatement]
The category $\SCatSim$ is a strict Cartesian closed category.
\end{theorem}

\begin{proof}
Consider the equipment defined in Section 6. 
To prove we have a Cartesian category, we show existence of a terminal object and existence and uniqueness of products. 
Terminality of the unit object (the empty type vector) follows from the terminality ($!$) equation. 
We define the pairing of terms $\term{N: ?r > !s}$ and $\term{M: ?r > !t}$ as
\begin{align*}
	[\term{N}, \term{M}] ~\defeq ~ \Delta ; (\term{N} \times \term{M}) = \term{<?x>.[!x].[!x] \ ; \ M.<?z>.N.[!z]} : \type{ ?r > !s!t} 
\end{align*}
To verify the existence of the product, we show that 
\begin{align*}
	[\term{N}, \term{M}] \ ; \pi_1 &=_{} \term{<?x>.[!x].[!x] \ ; \ M.<?z>.N.[!z] \ ;\  <?z'>.<?y>.[!y]} \\
		&=_\eta \term{<?x>.[!x].[!x].M.<?z>.N.[!z].<?z'>} \\
		&=_\beta \term{ <?x>.[!x].[!x].M.<?z>.N} \\
		&=_! \term{ <?x>.[!x].[!x].<?x`{'}>.N} \\
		&=_\beta \term{ <?x>.[!x].N} \\	
		&=_\eta \term{N} \\
	[\term{N}, \term{M}] \ ; \pi_2  &= \term{<?x>.[!x].[!x]\ ; \ M.<?z>.N.[!z]\ ; \ <?z'>.<?y>.[!z']} \\
		&=_\beta \term{<?x>.[!x].[!x].M.<?z>.N.<?y>.[!z]} \\
		&=_! \term{<?x>.[!x].[!x].M.<?z>.<?x'>.[!z]} \\
		&=_\iota \term{<?x>.[!x].[!x].<?x`{''}>.<?x'>.[!x`{''}].M} \\
		&=_\beta \term{<?x>.[!x].<?x'>.[!x].M} \\
		&=_\beta \term{<?x>.[!x].M} \\
		 &=_{\eta}  \term{M} 
\end{align*}
where $\term{!x, !x', !x`{''}: !r}$,  $\term{!y: !s}$ and $\term{!z, !z': !t}$.
To verify the uniqueness of the product, we show that for any term $\term{P: ?r > !s!t}$,
\begin{align*}
		\term{P} &=_{\beta\eta} \term{P.<?y>.<?z>.[!z].[!y].[!z].[!y].<?u>.<?v>.}\pi_1\term{.[!v].[!u].}\pi_2 \\
			&=_{\Delta} \term{<?x>.[!x].P.[!x].P.<?u>.<?v>.}\pi_1\term{.[!v].[!u].}\pi_2 \\
			&=_{\iota} \term{<?x>.[!x].P.}\pi_1\term{.[!x].P.}  \pi_2 \\
			&=_{\iota} \term{<?x>.[!x].[!x]  .  P.}\pi_2\term{.<?w>.P.}\pi_1\term{.[!w]} \\ 
			&= [\term{P} ; \pi_1, \term{P} ; \pi_2] 
\end{align*}
where $\pi_1 = \term{<?a>.<?b>.[!a]}$, $\pi_2 = \term{<?a>.<?b>.[!b]}$ and 
where $\term{x: !r}, \term{v, z: !s}$ and $\term{u, w, y: !t}$. 

For the closed structure, we show the existence and uniqueness of exponents. We define  the currying of a term $\term{M: ?s?r>!t}$ as 
\begin{align*}
	\term{M}^* \defeq_{\eta} \term{<?x>.[[!x].M]}: \type{?s > (?r > !t)} \ ,
\end{align*}
where $\term{!x: !s}$. 
To verify the existence and uniqueness, respectively, of exponents, we show that
\begin{align*} 
	  (\textsf{id}_{\cat{!r}} \times \term{M}^*)\  ;\  \epsilon &= \term{<?x>.[[!x].M]} \term{;} \term{<k>.k} =_\beta \term{<?x>.[!x].M} =_\eta \term{M} 
\\	 ((\textsf{id}_{\cat{!s}} \times \term{N}) \ ; \ \epsilon)^* &= \term{<?x>.[[!x].N.<k>.k]} =_{\epsilon}  \term{N} 
\end{align*}
for $\term{!x : !s}$ and $\term{k : ?r > !t}$.

The associators and unitors are  given by identities  and so are trivially natural isomorphisms which satisfy the coherence diagrams. This fact also makes the category a \textit{strict} Cartesian closed category. 
\end{proof}
%
%
Note that then the canonical morphisms can be recovered from the definition of products and currying above as
$	\Delta = \langle \textsf{id}, \textsf{id} \rangle, 	 \pi_1=\id\cat*{!},	 \pi_2={!}\cat*\id,	\term{M} \times \term{N} = \langle \pi_1 ; \term{M}, \pi_2 ; \term{N} \rangle,	\sigma = \langle \pi_2, \pi_1 \rangle , \eta = \textsf{id}^*$, \textit{etc}.
%

We proceed to show that the canonical functor induced by the soundness theorem is in fact faithful. 

\subsection{Completeness}

First, we introduce the simply typed $\lambda$-calculus with patterns, which we use for our translation. 
\begin{definition}
The \textit{$\lambda$-calculus (with $n$-ary tuples and patterns)} generated by a signature $\Sigma$ is given by the following grammar:
\begin{align*}
	M, N &~=~ x \mid~ MN ~\mid~ \lambda p.M ~\mid~ (M_1, \ldots, M_n) \\
	p,q, r, s, t&~=~ x~|~ (p_1, \ldots, p_n)
\end{align*}
where from left to right the \textit{term} constructors are a \textit{variable}, 
\textit{application}, \textit{abstraction} over a \textit{pattern} and \textit{product}. A {product} is a vector of terms and a pattern is a vector of variables. We freely allow coercion from patterns to terms. Terms are considered modulo $\alpha$-equivalence. We write products, contexts and patterns  as vectors $(s_1, \ldots, s_n)$ and elide the isomorphisms for associativity and unitality so that concatenation of $s$ and $t$ may be written as $(s \cdot t)$. 
\end{definition}

\begin{definition}
\textit{Simple types} are given by the following grammar\begin{align*}
A,B ~ \Coloneqq ~ \alpha \in \Sigma ~ \mid ~ {A} \to {B} ~\mid~ A \times B
\end{align*}
A \textit{judgement} $\Gamma \vdash M: A$ is a typed term in a \textit{context} $\Gamma = x_1:A_1, \ldots, x_n:A_n$, a finite function from variables to types. 
The \textit{typing rules} for the simply-typed $\lambda$-calculus (STLC) are given in Figure \ref{fig:STLC-types}.
\end{definition}

\begin{definition}
The \textit{equational theory} of the STLC is the least equivalence generated by the following laws, closed under any context:
\begin{align*}
	&\textup{Beta (Function):} & (\lambda p.M)N &\to_\beta M\{N/p\} \\
	&\textup{Eta (Function): } \textup{x} \not\in \fv{M}  &\lambda (p_1,\ldots, p_n).M(p_1,\ldots,p_n) &=_\eta M: A \to B \  \\
	&\textup{Eta (Product):} & (\pi_1(M), \ldots, \pi_n(M)) &=_\eta M: {A}_1 \times \ldots \times {A}_n
\end{align*}
where, in the first case, $\Gamma, x:A \vdash M:B$ and $\Gamma \vdash N:A$, and where
\[
	\{(N_1, \ldots, N_n)/(p_1, \ldots, p_n)\} = \{N_1 / p_1\}\ldots \{N_n / p_n\}
\]
denotes simultaneous substitution, and, in the last case, we define the syntactic sugar $\pi_i = \lambda (x_1, \ldots, x_n).x_i$. 
\end{definition}
\begin{figure}
\[
\begin{array}{cc}
\qquad
	\infer[\TR x]
	 {{p_1: {A_1} \ldots, p_n: {A_n} \vdash {p_i}: {A_i}}}
	{}
\end{array}
\]
\newline
\vspace{-\baselineskip}
\[
	\infer[ ]
	  {{\Gamma, (p_1, \ldots, p_n): {A_1} \times \ldots \times {A_n}  \vdash M : {C}}}
	  {{\Gamma, p_1: {A_1}, \ldots, p_n: {A_n} \vdash M: {C}}}
\qquad
\infer[\TR y]
	  {{\Gamma \vdash (M_1, \ldots, M_n): {A_1} \times \ldots \times {A_n}}}
	  {{\{\Gamma \vdash M_i: {A_i}}\}_{i\in \{1, \ldots, n\}} 
	  }
\]	
\newline
\vspace{-\baselineskip}
\[
\begin{array}{cc}\quad
	\infer[\TR a]
	  {{\Gamma \vdash MN: {B}}}
	  {{\Gamma \vdash M: {A}} &&
	   {\Gamma \vdash N: {A} \to {B}}
	  }
	& \qquad\qquad\quad
	\infer[\TR l]
	  {{\Gamma \vdash \lambda p. M : {A} \to {B}}}
	  {{\Gamma , p: {A} \vdash M: {B}}}
\end{array}
\]
\caption{Typing rules for the $\lambda$-calculus with patterns}
\label{fig:STLC-types}
\end{figure}
Note that in this calculus, the $\beta$-law for products is implemented by $\beta$-reduction.

We present the free functor induced by Theorem \ref{thm:sound-ccc-eqns} as a map from FMC terms to lambda terms, defined by induction on type derivations. Let $\textsf{CCC}(\Sigma)$ denote the free Cartesian closed category generated by a signature $\Sigma$ of base types.
\begin{lemma}
The interpretation functor $\termint{-}: \textup{\textsf{CCC}}({\Sigma}) \to \SCatSim$ can equivalently be defined inductively on $\lambda$-terms as follows. On types and contexts, respectively, define:
\begin{align*}
	\termint{{\alpha}} &= \type{\alpha}, \textup{ where } {\alpha} \in \Sigma \\
	\termint{{A} \to {B}} &=\termint{{A}} \type{>} \termint{{B}} \\
	\termint{A_1 \times \ldots \times A_n} &=\termint{A_1}\ldots\termint{A_n}  \\
	\termint{A_1, \ldots, A_n} &= \termint{A_n} \ldots \termint{A_1} 
\end{align*}
Inductively on the type derivation of $\Gamma \vdash M: A$, define an FMC term $\termint{\Gamma \vdash M: A} : \termint{\Gamma} \type{>} \termint{A}$
by
\begin{align*}
	\termint{ p_1: A_1, \ldots, p_n: A_n \vdash p_i: A_i} &= \term{<a_n>.\ldots.<a_1>.[a_i]},  \\
	\termint{\Gamma \vdash (M_1, \ldots, M_n): A_1 \times \ldots \times A_n} &= \term{<?x>.[!x].}\termint{\Gamma \vdash M_1: A_1}\term{\ldots} \term{[!x].}\termint{\Gamma \vdash M_n: A_n} \\ 
	\termint{\Gamma, (p_1, \ldots, p_n): A_1 \times \ldots \times A_n \vdash M:B} &= \termint{\Gamma, p_1: A_1, \ldots, p_n: A_n \vdash M:B}\\
	\termint{\Gamma \vdash MN: B} &= \term{<?x>.[!x].}\termint{\Gamma \vdash N: A}\term{.}\\
		& \qquad \term{[!x].}\termint{\Gamma \vdash M: A \to B}\term{.<k>.k} \\
	\termint{\Gamma \vdash \lambda p.M: A \to B} &= \term{<?x>.[[!x].}\termint{\Gamma, p: A \vdash M: B}\term{]}
\end{align*}
where in each case $\term{!x}: \termint{\Gamma}$, $\term{a_i} : \termint{A_i}$, and $\term{k}: \termint{A \to B}$.
\end{lemma}
\begin{proof}
We present below the equivalence between the simply typed $\lambda$-calculus and the free Cartesian closed category over the signature $\Sigma$. Using the functor from \textsf{CCC}($\Sigma$) to $\SCatSim$ described Theorem \ref{thm:ccc}  allows us to translate the right-hand sides to FMC terms, giving the result. 
In the following, $\langle M,N \rangle$ denotes the pairing of morphisms, given by $\Delta ; (M \times N)$. 
\begin{align*}
	\termint{ p_1: A_1, \ldots, p_n: A_n \vdash p_i: A_i} &= \pi_i,  \\ 
	\termint{\Gamma \vdash (M_1, \ldots, M_n): A_1 \times \ldots \times A_n} &= \langle \termint{\Gamma \vdash M_1: A_1}, \ldots, \termint{\Gamma \vdash M_n: A_n} \rangle \\
	\termint{\Gamma, (x_1, \ldots, x_n): A_1 \times \ldots \times A_n \vdash M:B} &= \termint{\Gamma, x_1: A_1, \ldots, x_n: A_n \vdash M:B}\\
	\termint{\Gamma \vdash MN: B} &= \langle \termint{\Gamma \vdash N: A}, \termint{\Gamma \vdash M: A \to B} \rangle \ ; \ \epsilon \\
	\termint{\Gamma \vdash \lambda p.M: A \to B} &= \termint{\Gamma, p: A \vdash M: B}^*\ . \qedhere
\end{align*}
\end{proof}


We construct a functor $\ccc{-}: \SCatSim \to \textup{\textsf{CCC}}(\Sigma)$ which we will show is the left-inverse of the interpretation $\termint{-}: \textup{\textsf{CCC}}({\Sigma}) \to \SCatSim$ given above. 
We first introduce a \textit{valuation}, which interprets the free variables of a term.
%
\begin{definition}[Valuation]
A \textit{valuation $v$} is a function assigning to each FMC-variable $\term{x:t}$ a $\lambda$-term $v(\term{x}) \in \ccctype{t}$. 
Given a valuation $v$, let $v\{\term{x}\leftarrow N\}$ denote the valuation which assigns $N$ to $\term x$ and otherwise behaves as $v$. 
\end{definition}

The definition of  $\ccc{-}: \SCatSim \to \textup{\textsf{CCC}}(\Sigma)$ follows the stack machine intuition of the FMC. The context of the corresponding lambda term represents the input stack and the lambda term itself represents the state of the stack at the end of the run, as a function of the input stack. 
The top-level arrow $\type{>}$ becomes sequent entailment $\vdash$ and the type of the input stack becomes the type of the $\lambda$-context and the type of the output stack becomes the type of the $\lambda$-term.

\begin{theorem}[Definition \ref{defn:interp-to-lambda} restatement]
For every signature $\Sigma$, define the interpretation $\ccc{-}: \SCatSim  \to \textup{\textsf{CCC}}(\Sigma) $ inductively on types as:
\begin{align*}
\ccc{\type{\alpha}} &= \alpha \textup{ where } \alpha \in \Sigma \\
\ccc{\type{t_1\ldots t_n}} &= \ccc{\type{t_1}} \times \ldots \times \ccc{\type{t_n}} \\
\ccc{\type{?s > !t}} &= \ccc{\type{!s}} \to \ccc{\type{!t}}\ .
\end{align*}
For each valuation $v$, define on the type derivation of $\term{G |- M: ?s > !t}$ an open $\lambda$-term
\[
	s: \ccctype{!s} \vdash \cccterm{G |- M: ?s > !t}_v(s): \ccctype{!t}
\]
  as follows:
\begin{align*}
\ccc{\Gamma \vdash \term{*: ?s > !s}}_v&( s)&=  &\  s \\
\ccc{\Gamma \vdash \term{x: \alpha}}_v&&=&\ v(\term{x})\\
\ccc{\Gamma \vdash \term{<x>.M: r?s > !t}}_v&( s \cdot r)&=  &\ \ccc{\Gamma, \term{x: r} \vdash \term{M: ?s > !t}}_{v\{\term{x} \leftarrow r\}} ( s) \\
\ccc{\Gamma \vdash \term{[N].M: ?s > !t}}_v&( s)&=  &\ \ccc{\Gamma \vdash \term{M: r?s >!t}}_v ( s \cdot \ccc{\Gamma \vdash \term{N: r}}_v) \\
\ccc{\Gamma, \term{x: ?r > !u} \vdash \term{x.M: ?r?s > !t}}_v&( s \cdot  r)&= &\ \ccc{\Gamma,\term{x: ?r > !u} \vdash \term{M: ?u?s > !t}}_v ( s \cdot v(\term{x})( r))\ .
\end{align*}
  We omit the context and/or types of terms inside the function when it is clear. 
\end{theorem}

We proceed to show that the interpretation functor is well-defined. 

We write $v\{\term{!x} \leftarrow  r\}$ as shorthand for $v\{\term{x_1} \leftarrow r_1 \}\ldots\{\term{x_n} \leftarrow r_n\}$, given a valuation $v$ and $r = (r_1, \ldots, r_n)$, and 
 $\cccterm{!{N}}_v$ as shorthand for $(\cccterm{N_1}_v \cdots \cccterm{N_n}_v)$. 
\begin{lemma}
For all terms $\term{G |- <?x>.M: ?r?s > !t}$, $\term{G |- !N: !r}$ and valuation $v$, we have that
\begin{align*}
	\cccterm{<?x>.M}_v( s \cdot  r) = \cccterm{M}_{v\{\term{!x} \leftarrow  r\}}( s)\quad {and} \quad
	\cccterm{[!N].M}_v( s) = \cccterm{M}_v( s \cdot \cccterm{!N}_v),
\end{align*} 
where ${r}: \ccctype{!r}$ and ${s} : \ccctype{!s}$.
\end{lemma}
\begin{proof}
By induction on the sizes of vectors $\term{!x}$ and $\term{!N}$, respectively.
\end{proof}

\begin{lemma}[Weakening]\label{weakening}
For all terms $\Gamma \vdash \term{M: t}$, valuations $v$ and $s \in \ccctype{s}$, we have that
\begin{align*}
\cccterm{G |- \term{M}}_{v} = \cccterm{G, \term{x} |- \term{M}}_{v\{\term{x} \leftarrow s\}} ,
\end{align*}
where $\term{x : s} \not\in \fv{\term{M}}$.
\end{lemma}
\begin{proof}
Induction on the type derivation of $\Gamma \vdash \term{M: t}$.
\end{proof}

\begin{lemma}[Sequencing]
\label{lem:ccc-sequencing}
For all terms $\Gamma \vdash \term{ M: ?u?s > !t}$ and $\Gamma \vdash \term{N: ?r > !u}$ and for all valuations $v$, we have that:
\begin{align*}
 \sn{\term{N;M}}_v ({s} \cdot {r}) &= \sn{\term{M}}_{v} ({s} \cdot  \sn{\term{N}}_{v} ({r}))
\end{align*}
\end{lemma}
\begin{proof}
We proceed by induction on the type derivation of $\Gamma \vdash \term{N: ?r > !u}$.
See Lemma \ref{lem:sequencing} for a similar proof. 
\end{proof}

\begin{lemma}[Substitution]\label{lem:ccc-substitution}
For every pair of terms $\Gamma, \term{x:w} \vdash \term{M: t}$ and $\Gamma \vdash \term{N: w}$ and for every valuation $v$, we have
\begin{align*}
	\cccterm{\term{\{N/x\}M}}_{v} = \cccterm{\term{M}}_{v\{\term{x} \leftarrow \sn{\term{N}}_v\} } \ .
\end{align*}
\end{lemma}
\begin{proof}
We proceed by induction on the type derivation of $\Gamma \vdash \term{M: t}$.
See Lemma \ref{lem:substitution} for a similar proof. 
\end{proof}

\begin{lemma}[Well-Definedness]
For any closed terms $\term{G |- M: t}$ and $\term{G |- N: t}$ and valuation $v$, we have
\begin{align*}
	\term{M} =_{\textsf{eqn}} \term{N} \quad {implies} \quad \cccterm{M}_v = \cccterm{N}_v \ . 
\end{align*}
\end{lemma}
\begin{proof}
We prove the statement in the case of each equation, making free use of the Sequencing and Weakening lemmas. 
\begin{itemize}
\item Beta: $\term{[N].<x>.M} \rw \term{M\{N/x\}: ?s > !t}$,
\begin{align*}
\cccterm{[N].<x>.M}_v( s) &= \cccterm{<x>.M}_v( s \cdot \cccterm{N}_v) = \cccterm{M}_{v\{\term{x} \leftarrow \cccterm{N}_v\}}(s) \\
	\textup{Substitution} &= \cccterm{\{N/x\}M}_v(s)
\end{align*}
where $\term{x, N:r}$.
\item 
Interchange: $\term{<?x>.N.[!x].M} =_\iota \term{M.<?y>.N.[!y]:  ?r?s > !t!u}$,
\begin{align*}
\cccterm{<?x>.N.[!x].M}_v( s \cdot  r) &= \cccterm{N.[!x].M}_{v\{\term{!x} \leftarrow  r\}}( s) \\
	&= \cccterm{[!x].M}_{v\{\term{!x} \leftarrow  r\}}(\cccterm{N}_v({s})) \\
	&= \cccterm{M}_{v}(\cccterm{N}_v({s}) \cdot  r)
\\ & = (\cccterm{N}_v({s}) \cdot \cccterm{M}_v( r)) \\
 \cccterm{M.<?y>.N.[!y]}_v({s} \cdot {r}) &= \cccterm{<?y>.g.[!y]}_v({s} \cdot \cccterm{M}_v({r})) \\
	&= \cccterm{N.[!y]}_{v\{\term{!y} \leftarrow \cccterm{M}_v({r}) \}}({s}) \\
	 &= \cccterm{[!y]}_{v\{\term{!y} \leftarrow \cccterm{M}_v({r}) \}}(\cccterm{N}_v({s})) \\
	&= (\cccterm{N}_v({s}) \cdot \cccterm{M}_v( r)) 
\end{align*}
where $\term{!x: !r}$, $\term{!y: !u}$, $\term{M: ?r > !u}$ and $\term{N: ?s > !t}$.
\item Diagonal: $\term{M.<?y>.[!y].[!y]} \to_\Delta \term{<?x>.[!x].M.[!x].M: ?s > !t!t}$,
\begin{align*}
 \cccterm{M.<?y>.[!y].[!y]}_v( s) &= \cccterm{<?y>.[!y].[!y]}_v(\cccterm{M}_v( s))  \\
	&= \cccterm{[!y].[!y]}_{v\{\term{!y} \leftarrow \cccterm{M}_v( s)\}}()  \\
	&= (\cccterm{M}_v( s)\cdot \cccterm{M}_v( s))  \\
\cccterm{<?x>.[!x].M.[!x].M}_v( s) &= \cccterm{[!x].M.[!x].M}_{v\{\term{!x} \leftarrow {s}\}}() \\
	&= \cccterm{M.[!x].M}_{v\{\term{!x} \leftarrow {s}\}}( s) \\
	&= \cccterm{[!x].M}_{v\{\term{!x} \leftarrow {s}\}}(\cccterm{M}_v( s)) \\
	&= \cccterm{M}_{v\{\term{!x} \leftarrow {s}\}}(\cccterm{M}_v( s) \cdot  s) \\
	&= (\cccterm{M}_v( s)\cdot \cccterm{M}_v( s))  
\end{align*}
where $\term{M: ?s > !t}$, $\term{!x: !s}$, $\term{!y: !t}$.
\item Terminal: $\term{M.<?y>} \to_! \term{<?x>: ?s > }$
\begin{align*}
	\cccterm{M.<?y>}_v( s) &= \cccterm{<?y>}_v(\cccterm{M}_v( s)) = (\epsilon) \\
	\cccterm{<?x>}_v(s) &= (\epsilon)
\end{align*}
where $\term{M : ?s > !t}$, $\term{!x: !s}$ and $\term{!y: !t}$
\item
{Extensionality:} $\term{M} = \term{<?x>.[[!x].M.<g>.g]: ?r > (?s > !t)}$
\begin{align*}
	\cccterm{<?x>.[[!x].M.<g>.g]}_v({r}) &= \cccterm{[[!x].M.<g>.g]}_{v\{\term{!x} \leftarrow {r}\}}(\epsilon) \\
	&= \cccterm{[!x].M.<g>.g}_{v\{\term{!x} \leftarrow {r}\}} \\
		&= \cccterm{M.<g>.g}_{v}({r}) \\
	&= \cccterm{<g>.g}_{v}(\cccterm{M}_v({r})) \\
		&= \cccterm{g}_{v\{\term{g} \leftarrow \cccterm{M}_v({r})\}}(\epsilon) \\
	&=\cccterm{M}_v({r}) 
\end{align*}
where $\term{M: ?r > (?s > !t)}$ $\term{!x: !r}$ and $\term{g: ?s > !t}$.
\end{itemize}
\end{proof}

\begin{lemma}
The interpretation $\ccc{-}: \SCatSim  \to \textup{\textsf{CCC}}(\Sigma)$ is a Cartesian closed functor. 
\end{lemma}
\begin{proof}
Functoriality follows immediately from Lemma \ref{lem:sequencing} (Sequencing). Unitality follows from observing $\ccc{\term{*}}_v = \textsf{id}$.
To show  $\ccc{-}$ is a \textit{Cartesian} functor, we must show it commutes with the product functor $\times$.
Making free use of the Sequencing and Weakening lemmas, we have:
\begin{align*}
\ccc{\term{M: ?s > !t} \times \term{N: ?r > !u}}_v( s \cdot  r) &= \cccterm{N.<?u>.M.[!u]}_v( s \cdot  r) \\
	&= \cccterm{<?u>.M.[!u]}_v( s\ \cdot \cccterm{N}_v( r)) &  \\
	 &= \cccterm{M.[!u]}_{v\{\term{!u} \leftarrow \cccterm{N}_v( r) \}}( s) \\
	&= \cccterm{[!u]}_{v\{\term{!u} \leftarrow \cccterm{N}_v(r) \}}(\cccterm{M}_{v\{\term{!u} \leftarrow \cccterm{N}_v( r) \}}( s)) & \\
	 &= \cccterm{[!u]}_{v\{\term{!u} \leftarrow \cccterm{N}_v( r) \}}(\cccterm{M}_{v}( s)) \\
	&= (\cccterm{M}_{v}( s) \cdot \cccterm{N}_v( r)) & \\
	&= (\cccterm{M}_v \times \cccterm{N})_v( s \cdot  r) \ ,
\end{align*}where $\term{!u: !u}$ and $s \in \ccc{\type{!s}}$,  $r \in \ccc{\type{!r}}$.
To show $\ccc{-}$ is a Cartesian \textit{closed} functor, we must further show that it preserves the arrow functor $\to$. Indeed, we have
\begin{align*}
\ccc{\term{M} \to \term{N}}_v(f) &= \cccterm{<x>.[M.x.N]}_v(f) \\
	&= \cccterm{[M.x.N]}_{v\{\term{x} \leftarrow f\}}(\epsilon) \\
	&= \cccterm{M.x.N}_{v\{\term{x} \leftarrow f\}} \\
	&=  \cccterm{M}_{v\{\term{x} \leftarrow f\}} \ ; \ \cccterm{x}_{v\{\term{x} \leftarrow f\}} \ ; \ \cccterm{N}_{v\{\term{x} \leftarrow f\}} \\
&=  \cccterm{M}_{v} \ ; \ f \ ; \ \cccterm{N}_{v} \\
	&= (\cccterm{M}_v \to \cccterm{N}_v)(f) 
\end{align*}
where $f : \ccc{\type{?s > !t}}$, $\term{x: ?s > !t}$ and $\term{M: ?r > !s}$ and $\term{N: ?t > !u}$.
\end{proof}

We demonstrate completeness by proving $\ccc{\termint{\Gamma \vdash M}}_v = \Gamma \vdash M: A$.
\begin{theorem}[Theorem \ref{thm:ccc-eqns} restatement]
The functor $\termint{-}: \textup{\textsf{CCC}}(\Sigma) \to \SCatSim$ is faithful.
\end{theorem}
\begin{proof}
We show that $\ccc{\termint{\Gamma \vdash M: A}}_v = \Gamma \vdash M: A$, or equivalently that 
\begin{align*}
	\ccc{\termint{a_1: A_1, \ldots, a_n: A_n \vdash M: B}}_v(a_n\cdots a_1) = M: B\ .
\end{align*}
Note we confuse between the input variables $a_i$ and variables of the context of $M$ in order to avoid proliferation of $\alpha$-conversions. 
We proceed by induction on the type derivation of $\Gamma \vdash M: A$. 
\begin{itemize}
\item The base (variable) case:
\begin{align*}
\ccc{\termint{a_1: A_1, \ldots, a_n: A_n \vdash a_i: A_i}}_v(a_n \cdots a_1) &= \ccc{\term{<a_n>.\ldots.<a_1>.[a_i]}}_v(a_n \cdots a_1) \\
	&= a_i,
\end{align*}
where $\term{a_i}: \termint{A_i}$.
\item The abstraction case:
\begin{align*}
	\ccc{\termint{{c}: \Gamma \vdash \lambda a.M: A \to B}}({c}) 
		&= \ccc{\term{<?x>.[[!x].}\termint{{c}, a\vdash M}\term{]}} ({c}) \\
		&= \ccc{\term{[[!x].}\termint{{c}, a \vdash M}\term{]}}_{\{\term{!x}\leftarrow {c}\}} (\epsilon) \\
		&= \ccc{\term{[!x].}\termint{{c}, a \vdash M}\term{}}_{\{\term{!x}\leftarrow {c}\}} \\
		\textup{$\eta$-expansion} &= \lambda a. \ccc{\term{[!x].}\termint{{c}, a \vdash M\term{}}}_{\{\term{!x}\leftarrow {c}\}}(a) \\
		&= \lambda a.\ccc{\termint{{c}, a \vdash M}\term{}}_{\{\term{!x}\leftarrow {c}\}} (a \cdot {c}) \\
	\textup{I.H.}	&= \lambda a.M
\end{align*}
where $\term{!x:} \termint{\Gamma}$
\item The application case:
\begin{align*}
	\ccc{\termint{{c: \Gamma} \vdash M(N): B}}({c})
		&= \ccc{\term{<?x>.[!x].}\termint{{c} \vdash N}\term{.[!x].}\termint{{c} \vdash M }\term{.<k>.k}}({c}) 
\\		&= \ccc{\termint{{c} \vdash N}\term{.[!x].}\termint{{c} \vdash M}\term{.<k>.k}}_{\{\term{!x} \leftarrow {c}\}}({c})
\\	\textup{Sequencing} &= \ccc{\term{[!x].}\termint{{c}\vdash M}\term{.<k>.k}}_{\{\term{!x} \leftarrow {c}\}}(\ccc{\termint{{c} \vdash N}}_{\{\term{!x} \leftarrow {c}\}} ({c}) )
\\	\textup{Weakening + I.H.} &= \ccc{\term{[!x].}\termint{{c} \vdash M}\term{.<k>.k}}_{\{\term{!x} \leftarrow {c}\}}( N )
\\	\textup{Sequencing} &= \ccc{\termint{{c} \vdash M}\term{.<k>.k}}_{\{\term{!x} \leftarrow {c}\}}( N \cdot c )
\\	\textup{Sequencing} &= \ccc{\term{<k>.k}}_{\{\term{!x} \leftarrow {c}\}}( N , \ccc{\termint{{c}\vdash M}}_{\{\term{!x} \leftarrow {c}\}} (c) )
\\	\textup{Weakening + I.H.} &= \ccc{\term{<k>.k}}( N , M )
\\	\textup{} &= \ccc{\term{k}}_{\{\term k \leftarrow M\}}( N )
\\	\textup{} &= M ( N )
\end{align*}
where $M: A \to B$, $N:A$, $\term{!x:} \termint{\Gamma}$ and $\term{k :} \termint{A \to B}$
\item The pattern case:
\begin{align*}
&\ccc{\termint{c: \Gamma, (p_1, \ldots, p_n): A_1 \times \ldots \times A_n \vdash M: B}}({c} \cdot p_1 \cdots p_n)  \\
&=\ccc{\termint{c, p_1, \ldots, p_n \vdash M}}({c} \cdot p_1 \cdots p_n) 
\\ \textup{I.H.} & = M
\end{align*}
\item The $n$-ary tuple case:
\begin{align*}
	\ccc{\termint{c : \Gamma \vdash (M_1,\ldots,M_n) }}({c}) &= \ccc{\term{<?x>.[!x].}\termint{c \vdash M_1}\term{\ldots[!x].}\termint{c \vdash M_n}} ({c})
\\			&=  \ccc{\term{[!x]}\termint{c \vdash M_1}\term{\ldots[!x].}\termint{c \vdash M_n}}_{\{\term{!x} \leftarrow {c}\}} ({\epsilon})
\\			&=  \ccc{\termint{c \vdash M_1}\term{\ldots[!x].}\termint{c \vdash M_n}}_{\{\term{!x} \leftarrow {c}\}} ({c})
\\	\textup{Sequencing} &=  \ccc{\term{\ldots[!x].}\termint{c \vdash M_n}}_{\{\term{!x} \leftarrow {c}\}} (\ccc{\termint{c \vdash M_1}}_{\{\term{!x} \leftarrow {c}\}}({c}))
\\	\textup{Weakening + IH} &=  \ccc{\term{\ldots[!x].}\termint{c \vdash M_n}}_{\{\term{!x} \leftarrow {c}\}} (M_1)
\\			&=  \ldots
\\			 &=  \ccc{\term{[!x].}\termint{c \vdash M_n}}_{\{\term{!x} \leftarrow {c}\}} (M_1 \cdots M_{n-1} )
\\			 &=  \ccc{\termint{c \vdash M_n}}_{\{\term{!x} \leftarrow {c}\}} (M_1 \cdots M_{n-1} \cdot {c})
\\	\textup{Sequencing}  &=   (M_1 \cdots M_{n-1} \cdot \ccc{\termint{c \vdash M_n}}_{\{\term{!x} \leftarrow {c}\}}({c}))
\\	\textup{Weakening + IH} &=   (M_1, \ldots, M_n)
\end{align*}
where $\term{!x: } \termint{\Gamma}$ and $M_i: A_i$. \qedhere
\end{itemize}
\end{proof}

We now show that terms modulo machine equivalence also form a Cartesian closed category. 
First, we show that machine equivalence is well-defined.

\begin{proposition}
\label{prop:congruence}
Machine equivalence $(\sim)$ is a congruence.
\end{proposition}

\begin{proof}
It is shown that $(\sim)$ is closed under contexts. Reflexivity then follows by induction on the term, symmetry is immediate, and transitivity follows by induction on types.
\begin{itemize}

	\item Unit case.
To show is that $\term{* \sim *:?t>!t}$. This requires that
\[
	\forall\, T\sim T':\type{!t}.~T\sim T':\type{!t}
\]
which is tautologous.

	\item Variable case.
To show is that if $\term{M\sim M':?s?t>!u}$ then $\term{x.M \sim x.M':?r?t>!u}$. Let both be typed in the context $\Gamma=\term{!w:!w}$, and assume this includes $\term{x:?r>!s}$. Let the following stacks be equivalent.
\[
	R\sim R':\type{!r} \qquad
	T\sim T':\type{!t} \qquad
	W\sim W':\type{!w}
\]
It must be shown that $U\sim U':\type{!u}$ for the following stacks $U$ and $U'$.
\[
\begin{aligned}
	U &=\result{TR}{\{W/!w\}x.M}
\\  U'&=\result{T'R'}{\{W'/!w\}x.M'}
\end{aligned}
\]
Let $W$ and $W'$ respectively contain the terms $\term N$ and $\term{N'}$ corresponding to $\term{x}$ in $\term{!w}$. Then $\term{\{W/!w\}x.M}=\term{N.\{W/!w\}M}$ and $\term{\{W/!w\}x.M\}}=\term{N'.\{W/!w\}M'}$. Let $\result RN=S$ and $\result{R'}{N'}=S'$, for which the assumption $\term{N\sim N'}$ gives $S\sim S':\type{!s}$. Then
\[
\begin{array}{l@{}c@{}r@{}c@{}r}
	U &~=~& \result{TR}{N.\{W/!w\}M} &~=~& \result{TS}{\{W/!w\}M}
\\
	U' &~=~& \result{T'R'}{N'.\{W'/!w\}M'} &~=~& \result{T'S'}{\{W'/!w\}M'}
\end{array}
\]
by composition of machine runs. By the premise $\term{M\sim M'}$ then $U\sim U':\type{!u}$ as required.

	\item Application case.
To show is that given $\term{N\sim N':r}$ and $\term{M\sim M':r?s>!t}$ then \\$\term{[N].M\sim[N'].M':?s>!t}$. Let all four terms be typed in the context $\Gamma=\term{!w:!w}$ and let the following stacks be equivalent.
\[
	S\sim S':\type{!s} \qquad
	W\sim W':\type{!w}
\]
It must be shown that $T\sim T':\type{!t}$ for the following stacks $T$ and $T'$.
\[
\begin{aligned}
	T &~=~\result{S}{\{W/!w\}[N].M}
\\	T'&~=~\result{S'}{\{W'/!w\}[N'].M'}
\end{aligned}
\]
For $T$, by the definition of substitution and by a single \emph{push} step of the machine,
\[
\begin{aligned}
	T &~=~\result{S}{[\{W/!w\}N].\{W/!w\}M}
\\	  &~=~\result{S\cdot\term{\{W/!w\}N}}{\{W/!w\}M}
\end{aligned}
\]
and similarly $T'=\result{S'\cdot\term{\{W'/!w\}N'}}{\{W'/!w\}M'}$. Since $\term{N\sim N'}$ then also \\$\term{\{W/!w\}N\sim \{W'/!w\}N'}$. By the premise $\term{M\sim M'}$ it follows that $T\sim T':\type{!t}$.

	\item Abstraction case.
To show is that if $\term{M\sim M':?s>!t}$ then it follows that \\$\term{<x>.M\sim<x>.M':r?s>!t}$. Let $\term{<x>.M}$ and $\term{<x>.M'}$ be typed in the context $\Gamma=\term{!w:!w}$ (not including $\term x$), and assume the following equivalent stacks and terms.
\[
	\term{N\sim N':r}\qquad
	S\sim S':\type{!s}\qquad
	W\sim W':\type{!w}
\]
It must be shown that $T\sim T':\type{!t}$ for the following stacks $T$ and $T'$.
\[
\begin{aligned}
	T &~=~\result{S\cdot\term{N}}{\{W/!w\}<x>.M}
\\	T'&~=~\result{S'\cdot\term{N'}}{\{W'/!w\}<x>.M'}
\end{aligned}
\]
By a single \emph{pop} step of the machine,
\[
\begin{aligned}
	T &~=~\result{S}{\{N/x,W/!w\}M}
\\	T'&~=~\result{S'}{\{N'/x,W'/!w\}M'}
\end{aligned}
\]
so that $\term{M\sim M'}$ gives $T\sim T':\type{!t}$. \qedhere
\end{itemize}
\end{proof}

\begin{proposition}
\label{prop:reduction-expansion-equivalence}
Machine equivalence includes $\beta\eta$-equivalence: $(=_{\beta\eta})\subseteq(\sim)$.
\end{proposition}

\begin{proof}
For beta, let $\term{[N].<x>.M}\rw\term{\{N/x\}M}$, and let both terms be typed in context $\term{!w:!w}$ with type $\type{?s>!t}$. Let $S\sim S':\type{!s}$ and $W\sim W':\type{!w}$. It must be shown that $T\sim T':\type{!t}$ for $T$ and $T'$ as follows.
\[
\begin{aligned}
	T &~=~\result{S}{\{W/!w\}[N].<x>.M}
\\	T'&~=~\result{S'}{\{W'/!w\}\{N/x\}M}
\end{aligned}
\]
Two steps of the machine on $\term{[N].<x>.M}$ evaluate the redex, to give the following.
\[
	T~=~\result{S}{\{W/!w\}\{N/x\}M}
\]
Then by reflexivity, $\term{M\sim M}$ and $\term{N\sim N}$, it follows that $T\sim T':\type{!t}$. Since $\sim$ is a congruence, if $\term M\rw\term N$ then $\term{M\sim N}$.

For expansion, let $(\term{M:r?s>!t})=_\eta(\term{<x>.[x].M:r?s>!t})$ where $x$ is not free in $\term M$, and let both terms be typed in the context $\term{!w:!w}$. Let $\term{N\sim N':r}$, $S\sim S':\type{?s}$, and $W\sim W':\type{!w}$. It must be shown that $T\sim T':\type{!t}$ for $T$ and $T'$ as follows.
\[
\begin{aligned}
	T &~=~\result{S\cdot\term{N}}{\{W/!w\}M}
\\	T'&~=~\result{S'\cdot\term{N'}}{\{W'/!w\}<x>.[x].M}
\end{aligned}
\]
Two steps of the machine on $\term{<x>.[x].M}$ pop then push back $\term{N'}$ on the stack, which gives the following.
\[
	T'~=~\result{S'\cdot\term{N'}}{\{W'/!w\}M}
\]
Then by reflexivity, $\term{M\sim M}$, it follows that $T\sim T':\type{!t}$. As with reduction, since $\sim$ is a congruence, if $\term M=_\eta\term N$ then $\term{M\sim N}$.
\end{proof}
We next establish that machine equivalence contains the equational theory of the FMC.
\begin{theorem}[Theorem \ref{thm:ccc} restatement]
For all typed, closed FMC terms $\term{M:t}$ and $\term{N: t}$, we have that
\begin{align*}
	\term{M} =_{\textsf{eqn}} \term{N}\ \qquad \textup{implies} \qquad  \term{M} \sim \term{N: t}. 
\end{align*}
\end{theorem}
\begin{proof}
We verify for each equation generating $=_\textsf{eqn}$, which suffices since by Proposition \ref{prop:congruence}, machine equivalence $\sim$ is closed under all contexts. In the following, we write $\term{!S}$ to represent the vector of terms in a stack $S$, and $\term{[!S]}$ to represent a corresponding sequence of applications (similar to the notation for $\term{[!x]}$). Beta and (first-order) eta-laws are given by the previous proposition.
\begin{itemize}
\item Interchange:
\[
	\term{M.<?x>.N.[!x]} =_\iota \term{<?y>.N.[!y].M} : \type{?r?t > !u!s}\ ,
\] 
 for $\term{M: ?r> !s}$, $\term{N: ?t > !u}$, $\term{!x}: \type{!s}$ and $\term{!y} : \type{r}$. We have that for all stacks $TR \sim T'R' : \type{!t!r}$,
\[
\begin{array}{@{}l@{~,}r@{}}
(~TR~& ~\term{M.<?x>.N.[!x]}~)
\\\hline
(~TS~& ~\term{<?x>.N.[!x]}~)
\\\hline
(~T~& ~\term{N.[!S]}~)
\\\hline
(~U~& ~\term{[!S]}~)
\\\hline
(~US~& ~\term{*}~)
\end{array}
\qquad
\textup{and}
\qquad
\begin{array}{@{}l@{~,}r@{}}
(~T'R'~& ~\term{<?y>.N.[!y].M}~)
\\\hline
(~T'~& ~\term{N.[!R'].M}~)
\\\hline
(~U'~& ~\term{[!R'].M}~)
\\\hline
(~U'R'~& ~\term{M}~)
\\\hline
(~U'S'~& ~\term{*}~)
\end{array}
\]
where $(R,\term M)\Downarrow S$, $(T, \term N)\Downarrow U$ and $(R', \term M)\Downarrow S'$, $(T', \term N)\Downarrow U'$. By reflexivity of $(\sim)$ we have that $S \sim S'$ and $U \sim U'$, and the result follows. 
\item Diagonal: 	
\[
	\term{M.<?y>.[!y].[!y]} =_\Delta \term{<?x>.[!x].M.[!x].M: ?s > !t!t},
\]
where $\term{M: ?s > !t}$, $\term{!x: !s}$ and $\term{!y: !t}$. We have that for all stacks $S \sim S': \type{!s}$, \[
\begin{array}{@{}l@{~,}r@{}}
(~S~& ~\term{M.<?y>.[!y].[!y]}~)
\\\hline
(~T~& ~\term{<?y>.[!y].[!y]}~)
\\\hline
(~\epsilon~& ~\term{[!T].[!T]}~)
\\\hline
(~T~& ~\term{[!T]}~)
\\\hline
(~TT~& ~\term{*}~)
\end{array}
\qquad
\textup{and}
\qquad
\begin{array}{@{}l@{~,}r@{}}
(~S'~& ~\term{<?x>.[!x].M.[!x].M}~)
\\\hline
(~\epsilon~& ~\term{[!S'].M.[!S'].M}~)
\\\hline
(~S'~& ~\term{M.[!S'].M}~)
\\\hline
(~T'~& ~\term{[!S'].M}~)
\\\hline
(~T'S'~& ~\term{M}~)
\\\hline
(~T'T'~& ~\term{*}~)
\end{array}
\]
where $(S, \term{M}) \Downarrow T$ and $(S', \term{M}) \Downarrow T'$. By reflexivity of $\sim$, we have that $T \sim T': \type{!t}$ and the result follows. 
\item Terminal:
\[
	\term{M.<?y>} =_! \term{<?x>: ?t > }
\]
where $\term{M: ?s > !t}$, $\term{!x:!t}$ and $\term{!y:!s}$. We have that for all stacks $T \sim T': \type{!s}$,
\[
\begin{array}{@{}l@{~,}r@{}}
(~T'~& ~\term{M.<?y>}~)
\\\hline
(~S~& ~\term{<?y>}~)
\\\hline
(~\epsilon~& ~\term{*}~)
\end{array}
\qquad
\textup{and}
\qquad
\begin{array}{@{}l@{~,}r@{}}
(~T'~& ~\term{<?x>}~)
\\\hline
(~\epsilon~& ~\term{*}~)
\end{array}
\]
where $(S, \term M) \Downarrow T$. The empty stack is trivially related to itself by $\sim$, and the result follows.
\item
Eta (Higher-order): 
\[
	\term{N} =_\epsilon \term{<?y>.[[!y].N.<x>.x]} : \type{?s > (?r > !t)} \ , 
\]
where $\term{N: ?s > (?r > !t)}$, $\term{x : ?r > !t}$ and $\term{!y : !s}$ .
We have that for all stacks $S \sim S': \type{!s}$,
\[
\begin{array}{@{}l@{~,}r@{}}
(~S~& ~\term{<?y>.[[!y].N.<x>.x]}~)
\\\hline
(~\epsilon~& ~\term{[[!S].N.<x>.x]}~)
\\\hline
(~\term{[!S].N.<x>.x}~& ~\term{*}~)
\end{array} \qquad
\textup{and}
\qquad
\begin{array}{@{}l@{~,}r@{}}
(~S'~& ~\term{N}~)
\\\hline
(~\term{P}~& ~\term{*}~)
\end{array}\, ,
\]
where $(S', \term{N}) \Downarrow \term{P}$. 
We thus require to show that $\term{[S].N.<x>.x} \sim \term{P} : \type{?r > !t}$. To verify this, observe that for all stacks $R \sim R' : \type{!r}$, 
\[
\begin{array}{@{}l@{~,}r@{}}
(~R~& ~\term{[!S].N.<x>.x}~)
\\\hline
(~R{S}~& ~\term{N.<x>.x}~)
\\\hline
(~R\cdot \term{P}~& ~\term{<x>.x}~)
\\\hline
(~R~& ~\term{P}~)
\\\hline
(~T~& ~\term{*}~)
\end{array}\qquad
\textup{and}
\qquad
\begin{array}{@{}l@{~,}r@{}}
(~R'~& ~\term{P}~)
\\\hline
(~{T'}~& ~\term{*}~)
\end{array}\, ,
\]
where $(S, \term N) \Downarrow \term{P'} : \type{?r > !t}$, $(R, \term{P}) \Downarrow T$ and $(R', \term{P}) \Downarrow T'$. By reflexivity, we have that $\term{P} \sim \term{P'}: \type{?r > !t}$ and thus $T \sim T': \type{!t}$, and the result follows.  \qedhere
\end{itemize} 
\end{proof}

\end{document}


References:

\cite{Church-1941}
\cite{Landin-1964} 
\cite{Landin-1965} 
\cite{Tait-1967} 
\cite{Plotkin-1975} 
\cite{Barendregt-1984}
\cite{Moggi-1991} 
\cite{Milner-Parrow-Walker-1992} 
\cite{Milner-1992} 
\cite{Parigot-1992} 
\cite{DeBruijn-1993}
\cite{Huet-1994} 
\cite{Hasegawa-1995} 
\cite{Takahashi-1995} 
\cite{Benton-Wadler-1996}
\cite{Filinski-1996} 
\cite{Power-Robinson-1997}
\cite{Pitts-Stark-1998} 
\cite{Streicher-Reus-1998} 
\cite{Maraist-Odersky-Turner-Wadler-1999} 
\cite{Power-Thielecke-1999}
\cite{Hughes-2000}
\cite{vonThun-2001} 
\cite{Plotkin-Power-2002-FOSSACS} 
\cite{Levy-2003} 
\cite{Levy-2006}
\cite{Lynas-Stoddart-2006} 
\cite{Krivine-2007} 
\cite{Plotkin-Pretnar-2009}
\cite{Lindley-Wadler-Yallop-2010} 
\cite{Pestov-Ehrenberg-Groff-2010} 
\cite{Atkey-2011} 
\cite{Ahman-Staton-2013-ENTCS} 
\cite{Egger-Mogelberg-Simpson-2014}
\cite{Ehrhard-Guerrieri-2016} 
\cite{Guerrieri-Manzonetto-2018} 
\cite{DalLago-Guerrieri-Heijltjes-2020}
\cite{Hirschkoff-Prebet-Sangiorgi-2020} 
